%% file: main.tex
\documentclass[10pt,a4paper]{article}
\usepackage{inputenc}
\usepackage[english]{babel}
\usepackage[nohead,includefoot,margin=2cm]{geometry}
\usepackage[compact,raggedright,small]{titlesec}
\usepackage[affil-it]{authblk}
\usepackage[runin]{abstract}
\usepackage{multicol}
\usepackage{graphicx}
\usepackage{hyperref} % add link to IBIResearch
\usepackage{enumitem} % use other counters lists

\input{math-definitions.tex}  

\input{acronyms} 

\usepackage{multirow}

\addto{\Affilfont}{\small}

\title{\large\bfseries Unique Compact Representation of Magnetic Fields using Truncated Solid Harmonic Expansions}
\date{\small \today}

\author[1,2]{Marija Boberg}
\author[1,2]{Tobias Knopp}
\author[1,2]{Martin M\"oddel}
\affil[1]{Section for Biomedical Imaging, University Medical Center Hamburg-Eppendorf, Hamburg, Germany}
\affil[2]{Institute for Biomedical Imaging, Hamburg University of Technology, Hamburg, Germany}

\begin{document}

  \maketitle

  \begin{abstract}
    Precise knowledge of magnetic fields is crucial in many medical imaging applications like \acl{MRI} or \ac{MPI} as they are the foundation of these imaging systems. For the investigation of the influence of field imperfections on imaging, a compact and unique representation of the magnetic fields using real solid spherical harmonics, which can be obtained by measuring a few points of the magnetic field only, is of great assistance. In this manuscript, we review real solid harmonic expansions as a general solution of Laplace's equation including an efficient calculation of their coefficients using spherical t-designs. We also provide a method to shift the reference point of an expansion by calculating the coefficients of the shifted expansion from the initial ones. These methods are used to obtain the magnetic fields of an \ac{MPI} system. Here, the \acl{FFP} of the spatial encoding field serves as unique expansion point. Lastly, we quantify the severity of the distortions of the static and dynamic fields in \ac{MPI} by analyzing the expansion coefficients.
  \end{abstract}
  
%% Keywords
 
\acresetall

\input{1_IntroAndProbStatement.tex}

%%%%%%%%%%%%%%%%%%
%%%%% Theory %%%%%
%%%%%%%%%%%%%%%%%%
\input{2_Theory}
    
%%%%%%%%%%%%%
%% Methods %%
%%%%%%%%%%%%%
\input{3_Methods}

%%%%%%%%%%%%%
%% Results %%
%%%%%%%%%%%%%
\input{4_Results}

%%%%%%%%%%%%%%%%
%% Discussion %%
%%%%%%%%%%%%%%%%
\input{5_Discussion}

\bibliographystyle{ieeetr}
\bibliography{MPI}

\newpage
\appendix
\section{Derivation of the Translation of the Coefficients}

\input{6_Appendix1}
\input{7_Appendix2}

\end{document}

%% file: math-definitions.tex
\usepackage{amsmath}
\usepackage{amsfonts}
\usepackage{amssymb}
\usepackage{nicefrac}
\usepackage{bm}  % better boldmath with hat
\usepackage[all]{xy}

\newcommand{\IN}{\mathbb{N}}		% natürliche Zahlen
\newcommand{\IR}{\mathbb{R}}		% reelle     Zahlen
\newcommand{\IS}{\mathbb{S}}		% Sphäre
\newcommand{\ball}{\mathcal{B}}		% Ball
\newcommand{\cC}{\mathcal{C}}       % continuous functions
\newcommand{\cS}{\mathcal{S}}       % specific function

\renewcommand{\epsilon}{\varepsilon}				% schöneres Epsilon
\renewcommand{\theta}{\vartheta}                  	% schöneres Theta
\renewcommand{\phi}{\varphi}						% schöneres Phi

\newcommand{\set}[1]{\left\{ #1 \right\}}				% Menge
\newcommand{\abs}[1]{\left\lvert #1 \right\rvert}		% Betrag
\newcommand{\norm}[1]{\left\lVert #1 \right\rVert}		% Norm
\newcommand{\scalar}[2]{\left\langle #1 , #2 \right\rangle} % Skalarprodukt
\newcommand{\onot}[1]{\mathcal{O}{\left( #1 \right)}}	% O-Notation

\newcommand{\signatur}[3]{#1:#2\,\rightarrow\,#3}     %Funktionen-Signatur

\newcommand{\kugel}{r,\theta,\phi}		% Kugelkoordinaten für Funktionsaufruf

\newcommand{\doppeq}{\mathrel{\mathop:}=}		% schöneres :=
\newcommand{\bxi}{ {\boldsymbol\xi} }

\newcommand{\bgam}{ {\boldsymbol\gamma} }

\newcommand{\brho}{ {\boldsymbol\rho} }

\newcommand{\DF}{\textup{DF}}
\newcommand{\SF}{\textup{SF}}
\newcommand{\FF}{\textup{FF}}
\renewcommand{\d}{\, \textup{d}}

\usepackage[exponent-product = \cdot]{siunitx}
\DeclareSIUnit{\mup}{\text{$\mu_0$}}

\usepackage{amsthm}

\theoremstyle{plain} % default: italic text, extra space above and below
\newtheorem{theorem}{Theorem}[section]
\newtheorem{lemma}[theorem]{Lemma}
\newtheorem{proposition}[theorem]{Proposition} 
\newtheorem{corollary}[theorem]{Corollary}

\theoremstyle{definition} % upright text, extra space above and below, upright bold caption
\newtheorem{definition}[theorem]{Definition}

\newtheorem*{remark}{Remark}
\newtheorem*{note}{Note}

\theoremstyle{remark} % upright text, no extra space above or below, italic caption

\usepackage{tikz}
\usepackage{tikz-3dplot}
\usepackage{pgfplots}
\usepackage{pgfplotstable}

\pgfplotsset{compat=1.8}

\usetikzlibrary{shapes,shapes.multipart,positioning}
\usepgfplotslibrary{units,statistics}

\pgfplotsset{
    unit markings=slash space,
    /pgfplots/xbar/.style={
    /pgf/bar shift={-0.5*(\numplotsofactualtype*\pgfplotbarwidth + (\numplotsofactualtype-1)*#1) + (.5+\plotnumofactualtype)*\pgfplotbarwidth + \plotnumofactualtype*#1},
    },
}
\pgfplotsset{select coords between index/.style 2 args={
    x filter/.code={
        \ifnum\coordindex<#1\fi
        \ifnum\coordindex>#2\fi
    }
}}

\usepackage{color}
\definecolor{ibilight}{RGB}{193,216,237}
\definecolor{ibidark}{RGB}{0,73,146}	% UKE-blau
\definecolor{uke2}{RGB}{170,156,143} 	% Mocca
\definecolor{uke3}{RGB}{87,87,86}		% Schwarz (Schrift)
\definecolor{ukesec1}{RGB}{255,223,0}	% Gelb
\definecolor{ukesec2}{RGB}{239,123,5}	% Orange (dunkel)
\definecolor{ukesec3}{RGB}{104,195,205}	% "TUHH"
\definecolor{ukesec4}{RGB}{138,189,36}	% Grün
\definecolor{ukesec5}{RGB}{178,34,41}	% Rot
\definecolor{tuhh}{RGB}{45,198,214}     % TUHH
\definecolor{ibidarkBG}{RGB}{227,229,242}   % UKE-blau
\definecolor{uke2BG}{RGB}{233,228,225} 	    % Mocca
\definecolor{uke3BG}{RGB}{230,231,232}	    % Schwarz (Schrift)
\definecolor{ukesec1BG}{RGB}{255,243,190}   % Gelb
\definecolor{ukesec2BG}{RGB}{254,232,212}   % Orange (dunkel)
\definecolor{ukesec3BG}{RGB}{222,241,241}   % "TUHH"
\definecolor{ukesec4BG}{RGB}{233,243,222}   % Grün
\definecolor{ukesec5BG}{RGB}{244,230,225}   % Rot

\usepackage{csvsimple}      % Table from csv data
\usepackage{arydshln}       % dotted lines

%% file: acronyms.tex
\usepackage[nolist,nohyperlinks]{acronym}

\begin{acronym}[ECU]

    \acro{MRI}[MRI]{magnetic resonance imaging}
    \acro{MPI}[MPI]{magnetic particle imaging}
    
    \acro{SPIOs}[SPIOs]{superparamagnetic iron oxide nanoparticles}

    \acro{FFP}[FFP]{field-free-point}
    \acro{FFL}[FFL]{field-free-line}
    \acro{FOV}[FOV]{field-of-view}
    \acro{LFR}[LFR]{low-field-region}
    
    \acro{SF}[SF]{selection field}
    \acro{DF}[DF]{drive field}
    \acro{FF}[FF]{focus field}
    
    \acro{ADC}[ADC]{analog-to-digital converter}
\end{acronym}

%% file: 1_IntroAndProbStatement.tex
\section{Introduction}
    Magnetic fields have been instrumental in the advancement of technology since the invention of the compass. They are fundamental for electric generators or motors, transformers, and magnetic storage devices. In the field of medical applications, magnetic fields are the basis of various imaging systems like \ac{MRI} or \ac{MPI}. The precise generation of the magnetic fields has a significant impact on image quality, as even small deviations can lead to image artifacts and misdiagnoses. If the magnetic fields are known, the negative influence of these imperfections can be corrected in most of these applications, like field-related artifacts in MR images~\cite{Janke2004}. A standard method for magnetic field representation is a spherical harmonic expansion, which can be obtained via a calibration measurement of the magnetic field at several positions on a spherical surface~\cite{ODonnell1987}. This offers a robust and compact representation of the distribution of said fields within a spherical region and allows analysis and solution of related problems~\cite{Eccles1993}.
    
    As with \ac{MRI}, the fundamental building blocks of the recent imaging modality \ac{MPI} are magnetic fields. Static magnetic fields spatially encode the \ac{MPI} signal while dynamic magnetic fields are used for signal generation~\cite{Rahmer2009}. \ac{MPI} scanners are characterized by the topology of their static signal encoding field, which is either a \ac{FFP} or a \ac{FFL}~\cite{Knopp2015a}. Many reconstruction methods in \ac{MPI} require some assumptions or knowledge about the magnetic fields. In x-space reconstruction the position of the \ac{FFP} is required to grid the measured data to the spatial domain in one step of the reconstruction~\cite{Goodwill2010}. During the fast implicit multi-patch reconstruction, which is used to increase the \ac{FOV}, the center positions of the different patches must be known to avoid artifacts~\cite{Szwargulski2018efficient}.
    
    While \ac{MPI} is used as application example in this paper, spherical harmonic expansions can be applied in various fields. As they provide a compact representation of magnetic fields, they are used in \ac{MRI} to effectively design active or passive shimming~\cite{Noguchi2014}, to determine the magnetic coupling between two electromagnetic sources~\cite{Van2014} or to model the earth's lithospheric magnetic field~\cite{Maus2006,Thebault2021}. Furthermore, spherical harmonics can be used for registration of objects by determination of the object's orientation~\cite{Burel1995} or for simulation of high-resolution, full-sky maps of the cosmic microwave background anisotropies~\cite{Muciaccia1997}.
    
    It was already shown by Bringout et al.~\cite{Bringout2014,Bringout2020} and Weber et al.~\cite{Weber2016,Weber2017} that the coefficients of a spherical harmonic expansion are suitable for the representation of magnetic fields in \ac{MPI}. In this paper, we give a review of \textit{real solid spherical harmonic expansions} as a general solution of Laplace's equation, which we will later apply to the magnetic fields in \ac{MPI}. For the calculation of the coefficients of this expansion, we use \textit{spherical t-designs} as efficient quadrature nodes. An actual measurement at these nodes provides the coefficients of the effective magnetic fields. The coefficients can be used directly to analyze the spatial characteristics of the magnetic fields at the point of the expansion. To exploit this, we present a \textit{method to shift the reference point of the expansion}. This offers the possibility to obtain the spatial characteristics of the magnetic fields at different positions from one set of coefficients calculated in a measurement based procedure. Finally, we use the coefficients at different expansion points for the \textit{characterization of static and dynamic fields in \ac{MPI}}.
    
    \subsection{Problem Statement}
    Magnetic particle imaging is a tracer based imaging modality, which determines the spatial distribution of \ac{SPIOs} using magnetic fields for signal generation and encoding. The signal encoding field is a static linear field, called selection field, and in the case considered here it has an \ac{FFP} topology. Only nanoparticles that are located inside a small \ac{LFR} around the \ac{FFP} are unsaturated and able to non-linearly respond to an excitation field, which leads to the spatial encoding of the signal. In our scenario, three orthogonal excitation directions with sinusoidal excitation and rational frequency ratios are chosen, such that the \ac{FFP} moves on a Lissajous trajectory running through a cuboidal \ac{FOV}. A detailed mathematical model of the MPI receive signal is described in~\cite{Kluth2018}. Here, we consider the signal under the assumptions of an ideal analog filter and no feed through. In this case the magnetization response of the nanoparticles in the \ac{LFR} induces the voltage signal into multiple receive coils $k=1,\dots,K$ given by
    \begin{align}\label{eq:MPISignalEquation}
        u_k(t) = 
        -\mu_0\int_{\Omega} \scalar{\bm p^k_\textup{r}(\bm r)}{\frac{\partial\bar{\bm m}}{\partial t}\!\left(\bm B(\bm r,t),t\right)}c(\bm r) \d\bm r,
    \end{align}
    where $\signatur{c}{\Omega}{\IR_{\geq 0}}$ describes the particle distribution inside the \ac{FOV} $\Omega\subseteq\IR^3$, $\signatur{\bar{\bm m}}{\IR^3\times\IR}{\IR^3}$ is the mean magnetic moment of the particles, $\signatur{\bm p^k_\textup{r}}{\IR^3}{\IR^3}$ is the coil sensitivity of the receive coils, and $\mu_0$ is the vacuum permeability. The magnetic moment of the particles is the response to the applied magnetic field $\signatur{\bm B}{\IR^3\times\IR}{\IR^3}$, which is composed of the static selection and dynamic drive fields.
    
    An exemplary \ac{MPI} experiment is shown in Fig.~\ref{fig:MPISetup}. A mouse is placed in the center of the scanner bore. During the measurement a tracer with \ac{SPIOs} is injected. In a measurement scenario typically selection fields with gradient strengths between \SIlist{0.2; 7}{\tesla\per\meter}~\cite{Graser2019human,Konkle2015convex} in $z$-direction with half of this value in $x$- and $y$-direction as well as drive-field amplitudes of about \SIrange{6}{18}{\milli\tesla}~\cite{Graser2019human,Weizenecker2009} are used. The drive-field amplitude is limited since higher amplitudes can cause peripheral nerve stimulation~\cite{Saritas2013}. Higher gradient strengths lead to a smaller signal generating \ac{LFR} such that the resolution of the imaging system increases~\cite{Rahmer2009}. However, this comes at the cost of a reduced size of the \ac{FOV}. E.g. using a gradient strength of \SI{2.0}{\tesla\per\m} with a drive-field amplitude of \SI{12}{\milli\tesla} yields a \ac{FOV} of \SI{24 x 24 x 12}{\milli\meter}, which does not cover larger objects like mice or rats. To this end a multi-patch approach is used~\cite{Knopp2015b}. Additional static magnetic fields, named focus fields, shift the initial \ac{FFP} such that different patches cover a larger \ac{FOV}. In Fig.~\ref{fig:MPISetup} an experiment is sketched, where a set of nine different patches is used to cover the mouse. 
    
    \begin{figure}
        \centering
        \includegraphics[width=0.8\textwidth]{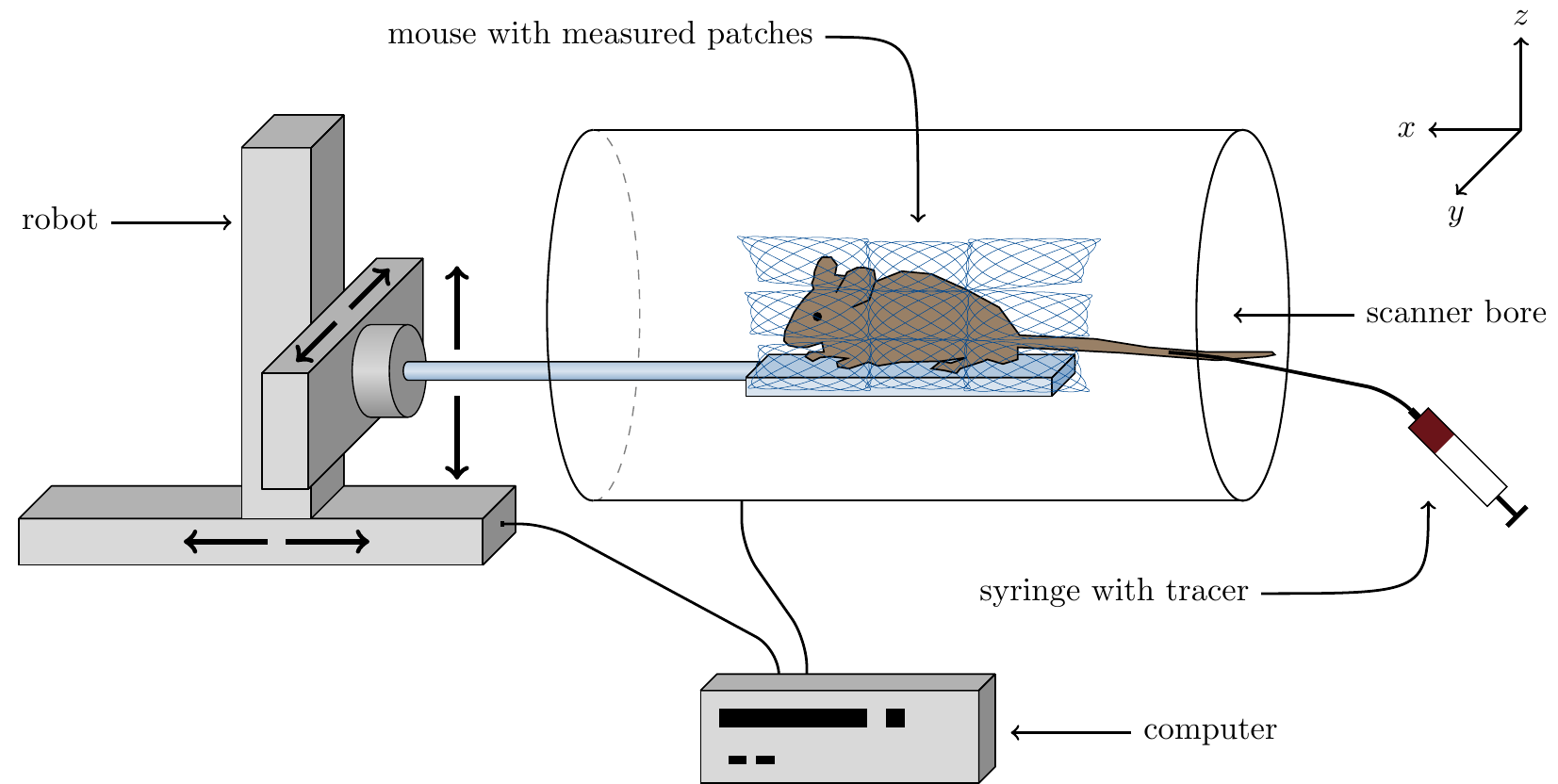}
        \caption{An MPI measurement is illustrated schematically. MPI scanner and a three-axis robot are controlled by a single computer. Prior to the measurement a mouse is placed in the center of the scanner bore using the robot. During the MPI measurement tracer material containing \ac{SPIOs} is injected into the mouse. As the size of the mouse exceeds the size of a single-patch \ac{FOV} multiple patches are used to cover its body. Off-center patches are warped due to the spatial characteristics of the static and dynamic fields.}
        \label{fig:MPISetup}
    \end{figure}
    
    Due to field imperfections the trajectory of each patch is slightly different, which might have multiple negative consequences. If the \ac{FFP} is not moving along the expected path the spatial encoding changes, which may lead to image artifacts. Moreover, patches might shift to different positions, which may lead to gaps the sampled \ac{FOV} like it is illustrated in Fig.~\ref{fig:MPISetup}. Lastly, different or spatially dependent drive-field amplitudes can result in incorrect estimations of the tracer concentration.
    
    One main goal of this work is to quantify the severity of the distortions of the underlying magnetic fields. Ideally, only constant and linear fields are present in \ac{MPI}, which, when represented by the coefficients of a spherical harmonic expansion at the \ac{FFP} of the selection field, leads to only a few non-zero coefficients as shown in Fig.~\ref{fig:IdealCoeffs}. Local imperfections are directly observable in non-zero coefficients of orthogonal field components or coefficients of higher order. For the calculation of the coefficients, we use field measurements at spherical t-design quadrature nodes located on a sphere. The corresponding expansion has the center of the sphere as its expansion point. However, since the exact position of the \ac{FFP} is not yet known at the time of measurement, this is likely not the \ac{FFP}. Additional measurement with the \ac{FFP} as center can be avoided by shifting the reference point of the expansion by a linear transformation of the coefficients. This method can also be used to compare the local fields at the centers of the patch positions, which ideally should be identical. This is done by moving the reference point of extension to each center, respectively.
    
    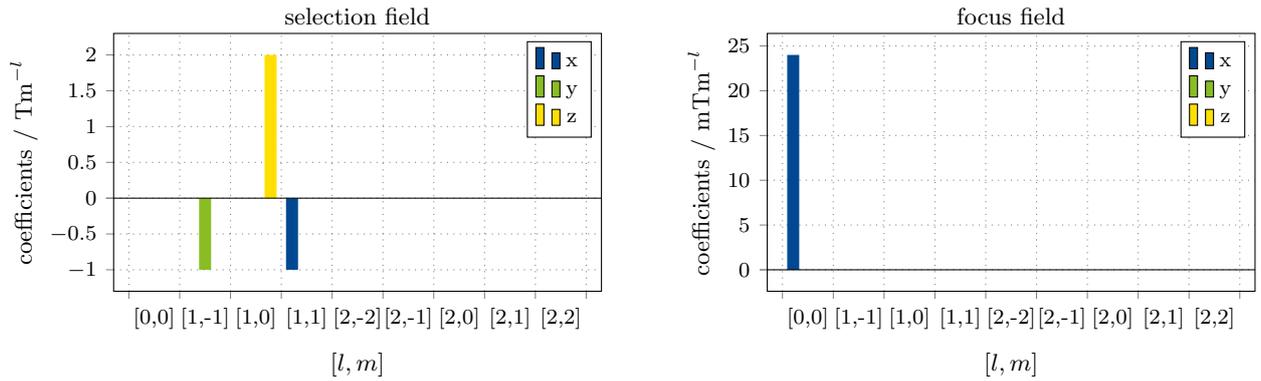
\begin{figure}
        \centering
        \input{tikz/Coeffs/Coeffs_IdealField.tex}
        \caption{Spherical harmonic coefficients of two ideal magnetic fields in \ac{MPI}. On the left, an ideal selection field with gradient strength of \SI{2}{\tesla\per\m} in $z$-direction and \SI{-1}{\tesla\per\meter} in $x$- and $y$-direction is shown. The gradient strength is represented by the linear coefficients ($l=1$) of the spherical harmonic expansion of the corresponding field direction. An ideal focus field in $x$-direction with a \SI{24}{\milli\tesla} field strength is visualized on the right. This constant field is represented by the constant coefficient ($l=0$) of the expansion in $x$-direction.}
        \label{fig:IdealCoeffs}
    \end{figure}

%% file: tikz/Coeffs/Coeffs_IdealField.tex
\pgfplotstableread[col sep=comma,]{tikz/Coeffs/data/Coeffs_IdealSF.csv}\datatable
\pgfplotstableread[col sep=semicolon,]{tikz/Coeffs/data/Labels_xticks.csv}\labelstable

\begin{tikzpicture}

\pgfplotsset{set layers=standard,
    footnotesize,
    width=8cm,
    height=5cm,
    grid=major,
    major grid style={white},
    %% labels %%
    change y base = true,
    xlabel={$[l,m]$},
    ylabel={coefficients},
    compat=1.3,
    xlabel shift={2pt},
    ylabel shift={-8pt},
    label style={font=\small},
    %% ticks %%
    tick label style={/pgf/number format/fixed,}, % einzelne konkrete ticks label
    tick style={major grid style={thin,dotted,gray}},
    tick align=outside,
    tickpos=left,
    xticklabels from table = {\labelstable}{lm},
    % label as interval for 3 bars each
            xtick={0.5,1.5,...,16.5},
            x tick label as interval,
            enlarge x limits=0.1,
    extra y ticks={0.0},
    extra x tick labels={},
    extra y tick labels={},
    extra x tick style={major grid style={thin,gray}},
    extra y tick style={major grid style={solid,black,on layer=axis foreground}},
    %% legend %%
    legend style = {anchor=north east, at={(0.98,0.97)}},
    legend columns=1,
    select coords between index={0}{8}
    }

%% Ideal selection field
\begin{axis}
    [name = SF, xshift=-4.3cm,
    title = selection field,
    y unit = Tm^{-\mathnormal{l}}, % mathnormal for italic l
    % bar plot
    ybar=1pt,
    bar width=4.5pt]
    
    \addplot[fill = ibidark,draw=none] table [x = num, y = x] {\datatable}; % x
    \addlegendentry{x}
    \addplot[fill=ukesec4,draw=none] table [x = num, y = y] {\datatable}; % y
    \addlegendentry{y}
    \addplot[fill=ukesec1,draw=none] table [x = num, y = z] {\datatable}; % z
    \addlegendentry{z}
\end{axis}

%% Ideal focus field
\pgfplotstableread[col sep=comma,]{tikz/Coeffs/data/Coeffs_IdealFF.csv}\datatable
\begin{axis}
    [name = DF, xshift = 4.3cm,
    title = focus field,
    y unit = Tm^{-\mathnormal{l}},
    y SI prefix = milli,
    % bar plot
    ybar=1pt,
    bar width=4.5pt]
    
    \addplot[fill = ibidark,draw=none] table [x = num, y = x] {\datatable}; % x
    \addlegendentry{x}
    \addplot[fill=ukesec4,draw=none] table [x = num, y = y] {\datatable}; % y
    \addlegendentry{y}
    \addplot[fill=ukesec1,draw=none] table [x = num, y = z] {\datatable}; % z
    \addlegendentry{z}
\end{axis}
\end{tikzpicture}

%% file: 2_Theory.tex
\section{Theory}

\subsection{Unique Solution of Laplace's Equation}

In this chapter, we start with the introduction of solid spherical harmonic expansions as general solution of Laplace's equation. In order to solve the equation, we use a Dirichlet boundary condition on a sphere, which is a natural choice for solutions expanded with spherical harmonics.

\begin{definition}
    Let $f\in\mathcal{C}^2(\ball_R(\brho),\IR)$ with $\ball_R(\brho)\doppeq\set{\bm a\in\IR^3 : \norm{\bm a - \brho}_2 \leq R}$, $\brho\in\IR^3$, $R\in\IR_+$. Laplace's equation with Dirichlet boundary condition is given by
    \begin{align}\label{eq:Th_LaplaceDirichlet}
        \begin{cases}
            \Delta f(\bm a) = 0\qquad &\forall \bm a\in\overset{\circ}{\ball}_R(\brho)\\
            f(\bm a) = \hat{f}(\bm a) \quad &\forall \bm a\in\partial\ball_R(\brho).
        \end{cases}
    \end{align}
    The boundary condition $\hat{f}\in\mathcal{C}(\partial\ball_R(\brho),\IR)$ is given on the surface of the ball denoted by $\partial\ball_R(\brho)$ while the interior is denoted by $\overset{\circ}{\ball}_R(\brho)$.
\end{definition}

First, we introduce a solution of \eqref{eq:Th_LaplaceDirichlet} on the unit ball, which we extend later for an arbitrary radius and center of the ball.

\begin{proposition}\label{Prop:CoeffsUnitSphere}
    Let $f\in\mathcal{C}^2(\ball_1(\bm 0),\IR)$ and $\signatur{\hat{f}}{\IS^2}{\IR}\in L^2(\IS^2)$ fulfill \eqref{eq:Th_LaplaceDirichlet}. Then, $f$ can be written as solid spherical harmonic expansion
    \begin{align}\label{eq:Th_solidExp}
        f(\bm a) = \sum_{l=0}^{\infty}\sum_{m=-l}^l \gamma_{l,m} Z_l^m(\bm a)\qquad \forall \bm a\in\ball_1(\bm 0)
    \end{align}
    where the normalized real solid spherical harmonics $Z_l^m$ as an extension of normalized real spherical harmonics are defined by
    \begin{align*}
            \signatur{Z_l^m}{\IR^3}{\IR}, (\kugel) \mapsto K_l^{\abs{m}} r^l P_l^{\abs{m}}(\cos \theta)\cdot
            \begin{cases}
                \sqrt{2}\cos (m\phi) & m > 0\\
                \sqrt{2}\sin (\abs{m}\phi) & m < 0\\
                1 & m = 0
            \end{cases}
    \end{align*}
    with $ K_l^m = \sqrt{\frac{(l-m)!}{(l+m)!}} $ and the associated Legendre polynomials $P_l^m$ \textup{\cite{Arfken2005}}.
    The solid spherical coefficients $\gamma_{l,m}\in\IR$ of the expansion can be calculated by the orthogonal projection
    \begin{align*}
        \gamma_{l,m} = \frac{2l+1}{4\pi}\int_{\IS^2}\hat{f}(\bm a)Z_l^m(\bm a) \d\bm a.
    \end{align*}
\end{proposition}

\begin{proof}\leavevmode
    The proposition holds since the restricted solid spherical harmonics $Z_l^m|_{\IS^2}$ form an orthogonal basis of $L^2(\IS^2)$ and thus the coefficients can be calculated by the orthogonal projection.
\end{proof}

\begin{remark}
    Each solid spherical harmonic expansion $\sum\limits_{l=0}^\infty \sum\limits_{m=-l}^l \gamma_{l,m} Z_l^m$ satisfies (2)~\cite{Arfken2005}.
\end{remark}

Next, we generalize proposition~\ref{Prop:CoeffsUnitSphere} for arbitrary radius $R$ and center $\brho$ of the ball $\ball_R(\brho)$. In this case, $\brho$ determines the center of the series expansion. 

\begin{proposition}\label{prop:SHE_arbBall}
    Let $f\in\cC^2(\ball_R(\brho),\IR)$ and $\signatur{\hat{f}}{\partial\ball_R(\brho)}{\IR}\in L^2(\partial\ball_R(\brho))$ fulfill eq.~\eqref{eq:Th_LaplaceDirichlet} for arbitrary $R\in\IR_+$ and $\brho\in\IR^3$. The coefficients $\gamma_{l,m}(\brho, R)$ depending on $\brho$ and $R$ can be calculated by
    \begin{align}\label{eq:Th_clmR}
        \gamma_{l,m}(\brho, R)
        = \frac{2l+1}{4\pi} \int_{\IS^2} \hat{f}{\left(R\bm a + \brho\right)}Z_l^m(\bm a) \d \bm a.
    \end{align}
    With $ \gamma_{l,m}(\brho) = \frac{1}{R^l}\gamma_{l,m}(\brho,R) $, the solid harmonic expansion of $f$ can be formulated as
    \begin{align}\label{eq:Th_ExpansionArbBall}
        f^{\brho}(\bm a) = 
        \sum_{l=0}^\infty\sum_{m=-l}^l \gamma_{l,m}(\brho) Z_l^m(\bm a)
        \qquad \forall \bm a\in\ball_R(\bm 0),
    \end{align}
    with $f^\brho := \tau_\brho(f)$ where $\signatur{\tau_\brho(f)}{\ball_R(\bm 0)}{\IR},\bm a \mapsto f(\bm a + \brho)$ denotes the shift operator on $\cC^2$ functions.
    The center $\brho$ determines the origin of the underlying coordinate system, which is denoted by a superscript $\brho$ for $f$. 
\end{proposition}

\begin{proof}
    Let $\signatur{f^{\brho,R} := \tau_\brho(\sigma_R(f))}{\ball_1(\bm 0)}{\IR},\bm a \mapsto f(R\bm a + \brho)$ and $\hat{f}^{\brho,R} := \tau_\brho(\sigma_R(\hat{f}))$ analogous, where $\sigma_R$ denotes the scaling operator on $\cC^2$ functions. Using proposition~\ref{Prop:CoeffsUnitSphere}, the coefficients of $f^{\brho,R}$ can be calculated by 
    \begin{align*}
        \gamma_{l,m}(\brho, R)
        = \frac{2l+1}{4\pi} \int_{\IS^2} \hat{f}^{\brho,R}{\left(\bm a\right)}Z_l^m(\bm a) \d \bm a,
    \end{align*}
    which yields~\eqref{eq:Th_clmR}. Thus, we get
    
    \begin{align*}
        f^{\brho,R} = \tau_\brho(\sigma_R(f)) &= \sum_{l=0}^{\infty}\sum_{m=-l}^l \gamma_{l,m}(\brho,R) Z_l^m\\
        \Leftrightarrow~
        \tau_\brho(f) = \tau_\brho(\sigma_R(\sigma^{-1}_R(f)))
        &= \sum_{l=0}^{\infty}\sum_{m=-l}^l \gamma_{l,m}(\brho,R)\,\sigma^{-1}_R(Z_l^m) \\
        &= \sum_{l=0}^{\infty}\sum_{m=-l}^l 
        \underbrace{\frac{\gamma_{l,m}(\brho,R)}{R^l}}_{=:\,\gamma_{l,m}(\brho)}
        Z_l^m
    \end{align*}
    since $Z_l^m$ are homogeneous polynomials so that $\sigma^{-1}_R(Z_l^m(\bm a)) = Z_l^m\!{\left(\frac{1}{R}\bm a\right)} = \frac{1}{R^l}Z_l^m(\bm a)$. With $f^\brho := \tau_\brho(f)$ eq.~\eqref{eq:Th_ExpansionArbBall} follows.
    
\end{proof}

\begin{note}
    For a better readability the indices $R$ and $\brho$ are omitted if $R=1$ and $\brho = \bm 0$.
\end{note}

%%%%%%%%%%%%%%%%%
%% Translation %%
%%%%%%%%%%%%%%%%%
\subsection{Translation}\label{ssec:Translation}

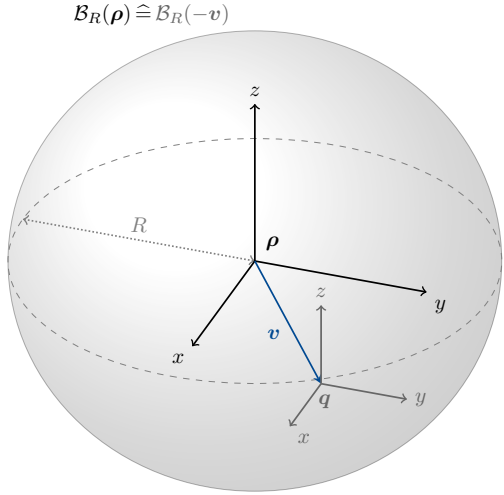
\begin{figure}
    \centering
    {\resizebox{8cm}{!}{\input{tikz/CoordinateSystems.tex}}}
    \caption{Different coordinate systems of the coefficients with the domain of the function $f$. The black coordinate system represents the initial coordinate system of the coefficients at the expansion point $\brho$. Using a shift $\bm v$, the coefficients depend on the shifted gray coordinate system with its origin at $\bm q = \brho + \bm v$. In this coordinate system the spherical domain of the function $f$ is now given by $\ball_R(-\bm v)$.}
    \label{fig:Th_CoordinateSystems}
\end{figure}

The coefficients $\gamma_{l,m}(\brho)$ correspond to a solid harmonic expansion around the center $\brho$ of the domain of the boundary condition. Next, we introduce a translation operator $\hat{\tau}$ that allows to transform them, such that they correspond to a solid harmonic expansion around a new center point $\bm q \in \ball_R(\brho)$. For instance, the coefficients $\gamma_{l,m}(\bm 0)$ can be calculated by $\hat{\tau}_{-\brho}{\big(\gamma_{l,m}(\brho)\big)}$ if $\bm 0 \in \ball_R(\brho)$.
The coordinate systems centered at the original and new expansion point are shown in Fig.~\ref{fig:Th_CoordinateSystems}. 

Henceforth, we assume $\signatur{f}{\ball_R(\brho)}{\IR}$ to be a polynomial of degree $L\in\IN_0$ and $\signatur{\hat{f}}{\partial\ball_R(\brho)}{\IR}\in L^2(\partial\ball_R(\brho))$ fulfill~\eqref{eq:Th_LaplaceDirichlet} for $R\in\IR_+$ and $\brho\in\IR^3$.

\begin{definition} We define the truncated solid harmonic expansion as a linear operator
    \begin{align*}
        \mathcal{S}_L : \IR^{(L+1)^2} &\rightarrow 
        \bigcup\limits_{\bm s\in \ball_R(\bm 0)}\hspace{-0.2cm}
        \mathcal{C}^2(\ball_R(\bm s),\IR),\\
        \bgam(\bm q) &\mapsto 
        \left(\bm a \mapsto  
        \sum_{l=0}^L\sum_{m=-l}^l \gamma_{l,m}(\bm q) Z_l^m(\bm a)\right),
    \end{align*}
    where $\bgam(\bm q) = \left(\gamma_{l,m}(\bm q)\right)_{\substack{l = 0,\dots,L \\ m = -l,\dots,l}} \in \IR^{(L+1)^2}$ is a vector containing all coefficients up to $l=L$ at expansion center $\bm q\in\ball_R(\brho)$. Since the domain of the expansion depends on the boundary condition used for the calculation of the coefficients, it holds that $ S_L(\bgam(\bm q)) \in \cC^2(\ball_R(\brho - \bm q),\IR) $.
\end{definition}

\begin{remark}\leavevmode
    \renewcommand{\labelenumi}{\textit{\roman{enumi})}}
    \begin{enumerate}
        \item Since we assume that $f$ is a polynomial of degree $L$,  \eqref{eq:Th_ExpansionArbBall} is equivalent to
        \begin{align*}
            f^\brho = \cS_L(\bgam(\brho))
        \end{align*}
        for $\bgam(\brho)$ calculated with \eqref{eq:Th_clmR}.
        
        \item A translation of the coordinate system by a shift $\bm v = \bm q - \brho$ to a new center $\bm q\in\ball_R(\brho)$ of the series expansion for $\bm a\in\ball_R(-\bm v)$ is described by
        \begin{align*}
            f^{\bm q}(\bm a) 
            = \tau_{\bm v}\big(f^{\brho}(\bm a)\big)
            &= \tau_{\bm v} \big(\mathcal{S}_L(\bgam(\brho))(\bm a)\big)\\
            &= \sum_{l=0}^L\sum_{m=-l}^l \gamma_{l,m}(\brho) \,\tau_{\bm v}\big(Z_l^m(\bm a)\big).
        \end{align*}
    \end{enumerate}
\end{remark}

The translation of the solid spherical harmonics can be transferred to the coefficients as it is stated in the following theorem.

\begin{theorem}\label{Th:Translation}
    For any $\bm q \in \ball_R(\bm \brho)$ an operator $\signatur{\hat{\tau}_{\bm v}}{\IR^{(L+1)^2}}{\IR^{(L+1)^2}}$ exists with $\bm v = \bm q - \brho$ such that 
    \begin{align*}
        \xymatrix{
            \IR^{(L+1)^2}  
            \ar@{-->}[r]^{\hat{\tau}_{\bm v}} \ar[d]_{\mathcal{S}_L} 
            & \IR^{(L+1)^2} \ar[d]_{\mathcal{S}_L}\\
           \mathcal{C}^2(\ball_R(\brho),\IR) 
           \ar[r]^{\tau_{\bm v}} 
           & \mathcal{C}^2(\ball_R(- \bm v),\IR)  
        }
    \end{align*}
    commutes, i.e.
    \begin{align*}
        \tau_{\bm v} \circ \mathcal{S}_L(\bgam(\brho)) 
        = \mathcal{S}_L \circ \hat{\tau}_{\bm v}\big(\bgam(\brho)\big).
    \end{align*}
    The actual calculation of $\hat{\tau}_{\bm v}\big(\bgam(\brho)\big) = \bgam(\bm q)$ is given by equations \eqref{eq:Coeffs_m>0} to \eqref{eq:Coeffs_l=0} in appendix \ref{ssec:Ap_AdditionCoefficients}.
\end{theorem}

\begin{proof}
    The proof of the theorem is given in the appendix. First, we adapt the addition theorem for unnormalized real solid spherical harmonics provided by Rico et al. in their work~\cite{Rico2013} to our normalized ones in section~\ref{ssec:Ap_AdditionHarmonics}. Applying the addition theorem to the solid harmonic expansion and reordering of the sums leads to the addition theorem for the solid coefficients and with that to the proof of theorem as it is shown in the second section~\ref{ssec:Ap_AdditionCoefficients}.
\end{proof}

\begin{remark}\leavevmode
    \renewcommand{\labelenumi}{\textit{\roman{enumi})}}
    \begin{enumerate}
        \item The effort for calculating $(L+1)^2$ coefficients $\hat{\tau}_{\bm v}\big(\bgam(\brho)\big)$ with \eqref{eq:Coeffs_m>0} to \eqref{eq:Coeffs_l=0} is $\onot{L^4}$.
    
        \item Theorem~\ref{Th:Translation} can be generalized for a shift between arbitrary points $\bm p, \bm q \in \ball_R(\brho)$.
    \end{enumerate}
\end{remark}
    
\subsection{Efficient Quadrature}\label{ssec:Quadrature}
In order to obtain the solid coefficients for a polynomial $\signatur{f}{\ball_R(\brho)}{\IR}$ of degree $L\in\IN_0$ its values on the boundary $\partial\ball_R(\brho)$ have to be known. For instance in the magnetic field determination application scenario, these values can be measured. In MPI, for example, the choice of the measurement points is only restricted by the size and shape of the scanner bore and so far a Gauss-Legendre quadrature was used in \ac{MPI} for the calculation of the coefficients~\cite{Bringout2020}. A more efficient way to choose the measurement points are spherical t-designs~\cite{Beentjes2015}, which are introduced next.

%%%% t-design
\begin{definition}\label{Def_tDesign}
    A \textit{spherical t-design} is a set of nodes $\set{\bm a_k}_{k=1,\dots,N}\subseteq \IS^2$ such that
    \begin{align*}
        \int_{\IS^2} \mathcal{Y}(\bm a) \d\bm a = \frac{4\pi}{N}\sum_{k=1}^{N} \mathcal{Y}(\bm a_k)\quad\forall \mathcal{Y} \in \Pi^t
    \end{align*}
    with $\Pi^t = \text{span}{\set{Z_l^m|_{\IS^2} : l \leq t}}$ being the set of all polynomials up to degree $t\in\IN$ on the unit sphere.~\cite{Beentjes2015}
\end{definition}   

\begin{remark}
    The spherical t-design is a very efficient sampling pattern for the quadrature on a spherical surface. Multiple t-designs can be found in~\cite{Hardin8Design}. For instance, the smallest known $8$-design only consists of $36$ points~\cite{Hardin1996}. In comparison, $45$ Gauss-Legendre quadrature nodes are required for the same accuracy~\cite{Weber2017}.
\end{remark}

\begin{proposition}\label{Prop:CoeffsTDesign}
    Assume $\signatur{f}{\ball_R(\brho)}{\IR}$ to be a polynomial of degree $L \in \IN_0$ fulfilling  \eqref{eq:Th_LaplaceDirichlet} with $\signatur{\hat{f}}{\partial\ball_R(\brho)}{\IR}\in L^2(\partial\ball_R(\brho))$.
    Let $\set{\bm a_k}_{k=1,\dots,N}$ be a $2L$-design. For $l\leq L$ it holds that
        \begin{align*}
            \gamma_{l,m}(\brho)
            = \frac{2l+1}{N R^l}\sum_{k=1}^N \hat{f}(R\bm a_k+\brho)Z_l^m(\bm a_k).
        \end{align*}
\end{proposition}

\begin{proof}
        Let $l\leq L$. Then, it holds that the degree of the product $\hat{f} Z_l^m$ is at most $2L$. 
        With proposition~\ref{prop:SHE_arbBall} and definition~\ref{Def_tDesign}, we get
        \begin{align*}
            \gamma_{l,m}(\brho)
            &= \frac{1}{R^l}\frac{2l+1}{4\pi} \int_{\IS^2}\hat{f}(R\bm a+\brho)Z_l^m(\bm a) \d\bm a \\
            &= \frac{1}{R^l}\frac{2l+1}{4\pi}\frac{4\pi}{N}\sum_{k=1}^N \hat{f}(R\bm a_k+\brho)Z_l^m(\bm a_k).
        \end{align*}
        Since $\deg(f) \leq L$ the solid harmonic expansion~\eqref{eq:Th_ExpansionArbBall} can be truncated at $l = L$.
\end{proof}

%% file: tikz/CoordinateSystems.tex
\tdplotsetmaincoords{60}{110}

%define polar coordinates for some vector
\pgfmathsetmacro{\rvec}{1.4}
\pgfmathsetmacro{\thetavec}{90}
\pgfmathsetmacro{\phivec}{35}

%start tikz picture, and use the tdplot_main_coords style to implement the display 
%coordinate transformation provided by 3dplot
\begin{tikzpicture}[scale=3,tdplot_main_coords]

%set up some coordinates 
%-----------------------
\coordinate (O) at (0,0,0);

\tdplotsetcoord{P}{\rvec}{\thetavec}{\phivec}
\tdplotsetcoord{Q}{\rvec-0.2}{\thetavec}{\phivec+135}

%draw figure contents
%--------------------
% measurement ball
\tdplotsetrotatedcoordsorigin{(O)}
\tdplotsetrotatedcoords{20}{80}{0}
\draw [ball color=white,very thin,tdplot_rotated_coords,opacity=0.3] (0,0,0) circle (1.35); % ball
\draw [dashed,gray] (0,0,0) circle (1.35); % equator
\draw[thick,<->,color=gray,densely dotted] (O) -- (0,-1.35,0) node[midway,above] {$R$}; % radius
\node[gray] at (0,-0.6,1.45) {$\textcolor{black}{\ball_R(\brho)\,\widehat{=}\, } \textcolor{black!60}{\ball_R(-\bm v)}$}; % label

%draw the main coordinate system axes
\draw[thick,->] (0,0,0) -- (1,0,0) node[anchor=north east]{$x$};
\draw[thick,->] (0,0,0) -- (0,1,0) node[anchor=north west]{$y$};
\draw[thick,->] (0,0,0) -- (0,0,1) node[anchor=south]{$z$};

\node[above right = 0.08cm] {$\brho$};

%draw a vector from origin to point (P) 
\draw[-stealth,color=ibidark,thick] (O) -- (P) node[midway, below left=-0.01cm]{$\bm v$};

%set the rotated coordinate definition within display using a translation
%coordinate and Euler angles in the "z(\alpha)y(\beta)z(\gamma)" euler rotation convention
%syntax: \tdplotsetrotatedcoords{\alpha}{\beta}{\gamma}
\tdplotsetrotatedcoords{0}{0}{0}

%%%%% P %%%%%%
%translate the rotated coordinate system
%syntax: \tdplotsetrotatedcoordsorigin{point}
\tdplotsetrotatedcoordsorigin{(P)}

%use the tdplot_rotated_coords style to work in the rotated, translated coordinate frame
\draw[thick,tdplot_rotated_coords,->,color=black!60] (0,0,0) -- (.5,0,0) node[anchor=north west]{$x$};
\draw[thick,tdplot_rotated_coords,->,color=black!60] (0,0,0) -- (0,.5,0) node[anchor=west]{$y$};
\draw[thick,tdplot_rotated_coords,->,color=black!60] (0,0,0) -- (0,0,.5) node[anchor=south]{$z$};

\node[below = 0.06cm, xshift=0.06cm, color=black!60] at (P) {$\bm q$};

\end{tikzpicture}

%% file: 3_Methods.tex
\section{Methods}\label{sec:Methods}

%%%%%%%%%%%%%%%%%%%%%
%% Magnetic Fields %%
%%%%%%%%%%%%%%%%%%%%%
\subsection{Magnetic Fields}

In \ac{MPI}, we are interested in the quasi-static magnetic field inside the scanner bore, which are generated by electric currents outside the bore. As these fields fulfill Laplace's equation, we are able to apply proposition~\ref{prop:SHE_arbBall} and expand the field as a solid harmonic expansion.

\begin{lemma}\label{lem:MagnField_Laplace}
    Let $\bm B = (B_x,B_y,B_z) \in\mathcal{C}^2(\Omega,\IR^3)$ be a quasi-static magnetic field, where $\Omega\subseteq\IR^3$ describes the region where no electric current flows. Then, the components $B_i$ for $i=x,y,z$ fulfill Laplace's equation for all $\bm a \in \Omega$.
\end{lemma}

\begin{proof}
    The static magnetic field fulfills Maxwell's equations 
    \begin{align*}
        \nabla\cdot\bm B(\bm a) &= 0\\
        \nabla\times\bm B(\bm a) &= \bm 0
    \end{align*}
    for all $\bm a\in\Omega$~\cite{Jackson1999}. The second equation is equal to zero since no electric current flows in $\Omega$. Thus, it holds that $\bm B = -\nabla\Phi$ where $\signatur{\Phi}{\Omega}{\IR}$ is the magnetic scalar potential \cite{Weber2017, Jackson1999}. In \cite{Weber2017} it is shown that $-\Delta \Phi = 0$ follows, which implies that $\Delta B_i = 0$.
\end{proof}

\begin{note}
    In the following, the index $i$ denotes the $x$-, $y$-, and $z$-component of the magnetic field $\bm B$.
\end{note}

\begin{proposition}\label{Prop:MagnField_SphExp}
    Assume that the magnetic field $\bm B \in\mathcal{C}^2(\Omega,\IR^3)$ can be described by a polynomial of degree $L$. Let $R\in\IR_+$ and $\brho\in\IR^3$ such that $\ball_R(\brho)\subseteq\Omega$ and let $\set{\bm a_k}_{k=1,\dots,N}$ be a $2L$-design. With a given boundary condition $B_i(\bm a) = \hat{B}_i(\bm a) ~ \forall \bm a\in\partial\ball_R(\brho)$ with $\hat{B}_i\in L^2{\left(\ball_R(\brho)\right)}$ the field can be formulated as a solid harmonic expansion
    \begin{align*}
        \bm B^{\brho}(\bm a) = 
        \sum_{l=0}^L\sum_{m=-l}^l \bgam_{l,m}(\brho) Z_l^m(\bm a)
        \qquad \forall \bm a\in\ball_R(\bm 0)
    \end{align*}
    with coefficients $\bgam_{l,m} = (\gamma_{l,m}^x, \gamma_{l,m}^y, \gamma_{l,m}^z)$ calculated by
    \begin{align*}
        \gamma_{l,m}^i(\brho)
        = \frac{1}{R^l}\frac{2l+1}{N} \sum_{k=1}^N \hat{B_i}{\left(R\bm a_k + \brho\right)}Z_l^m(\bm a_k).
    \end{align*}

\end{proposition}
    
\begin{proof}
    Due to lemma~\ref{lem:MagnField_Laplace} the magnetic field $\bm B$ and the boundary $\hat{\bm B}$ fulfill all assumptions of proposition~\ref{Prop:CoeffsTDesign}, which implies both equations.
\end{proof}

    The expansion of the magnetic field $\bm B$ at the expansion point $\brho$ with coefficients coefficients $\bgam^x(\brho)$, $\bgam^y(\brho)$, and $\bgam^z(\brho)$ characterizes the magnetic field locally in $x$-, $y$-, and $z$-direction, respectively. Similar to Taylor series, the expansion can be written as polynomial where the polynomial degree increases with the index $l$. I.e. $\bgam_{0,0}$ describes the constant part of the magnetic field, while the coefficients $\bgam_{1,m}$ contain the information about its linear behavior, $\bgam_{1,-1}$ describes the behavior in $y$-direction, $\bgam_{1,0}$ in $z$-direction, and $\bgam_{1,1}$ in $x$-direction. The coefficients for $l>1$ characterize the nonlinear behavior of the magnetic field.

%%%%%%%%%%%%%%%%%%%%%%%%%%%%
%% Magnetic fields in MPI %%
%%%%%%%%%%%%%%%%%%%%%%%%%%%%
\subsection{Magnetic Fields in MPI}
In \ac{MPI}, two main magnetic fields are used for signal encoding and generation: a linear selection field $\signatur{\bm B_\SF}{\IR^3}{\IR^3}$ and dynamic drive fields $\signatur{\bm B_\DF^i}{\IR^3\times\IR}{\IR^3}$. Each dynamic drive field can be separated into the constant coil sensitivity $\signatur{\bm p_\DF^i}{\IR^3}{\IR^3}$ and the sinusoidal current $\signatur{I^i}{\IR}{\IR}$ such that $\bm B_\DF^i(\bm r,t) = I^i(t)\bm p_\DF^i(\bm r)$. In our setup, the three orthogonal drive-field coils are also the receive coils so that $\bm p_\DF^k = \bm p^k_\textup{r}$ with $k\in\set{1,2,3}$ in~\eqref{eq:MPISignalEquation}. In the multi-patch setting described in the problem statement additional patch-wise constant focus fields $\signatur{\bm B_\FF^i}{\IR^3}{\IR^3}$ are applied to obtain a larger \ac{FOV}. More general information on the imaging principles of \ac{MPI} and the setup of an \ac{MPI} scanner can be found in~\cite{Knopp2012b, Knopp2015a} while the mathematical background of \ac{MPI} is described in~\cite{Kluth2018}.

\begin{table}[ht]
        \renewcommand{\arraystretch}{1.2}
        \centering
        \begin{tabular}{|c||c|c|c|}
            \hline
            & \textbf{Selection Field} & \textbf{Focus Field} & \textbf{Drive Field} \\
            \hline\hline
            $l=0$ & $\bgam_{0,0} = \bm 0$ 
            & $\gamma^i_{0,0} = f^i$ 
            & $\gamma^i_{0,0} = d^i$ \\
            \hdashline
            & $\gamma^y_{1,-1} = -0.5g$ 
            & & \\
            $l=1$ & $\gamma^z_{1,0} = g$ 
            & $\bgam_{1,m} = \bm 0$ 
            & $\bgam_{1,m} = \bm 0$\\
            & $\gamma^x_{1,1} = -0.5g$ 
            & & \\
            \hline
        \end{tabular}
        \caption{Coefficients of the three different ideal magnetic fields in MPI in \si{\tesla\per\meter\tothe{\mathnormal{l}}}.}
        \label{tab:IdealCoeffs}
        \renewcommand{\arraystretch}{1.0}
\end{table}

The solid coefficients at the expansion point $\bm 0$ of ideal selection, focus and drive fields are listed in Table~\ref{tab:IdealCoeffs}. All coefficients not mentioned in the table are zero. Since the selection field is a linear field only the coefficients for $l=1$ are nonzero and describe the gradient strength of $g\in [\SI{0}{\tesla\per\m},\SI{2.5}{\tesla\per\m}]$ in $z$-direction and $-0.5g$ in $x$- and $y$-direction. Meanwhile only the constant coefficients for $l=0$ of the ideal drive and focus field are nonzero. The focus fields are characterized by the shift $f^x,f^y\in [\SI{-17}{\milli\tesla},\SI{17}{\milli\tesla}]$ and $f^z\in [\SI{-42}{\milli\tesla},\SI{42}{\milli\tesla}]$, while the drive field is characterized by its amplitude $d^i\in [\SI{0}{\milli\tesla},\SI{14}{\milli\tesla}]$ with $d^i = \max_t(I^i(t)) \bm p_\DF^i(\bm 0)$.

%%%%%%%%%%%%%%%%%%%%%%%
%% Measurement setup %%
%%%%%%%%%%%%%%%%%%%%%%%
\subsection{Measurement Setup}\label{ssec:MeasurementSetup}
Measuring magnetic fields can be done using different devices like Hall-effect sensors, SQUID sensors or induction sensors~\cite{Tumanski2013}. Gaussmeters with a 3-axis Hall sensor are very accurate and therefore widely used for magnetic field measurements~\cite{Renella2017overview, Renella2019}. Hence, we use a 3-channel gaussmeter with a three-axis high-sensitivity Hall-effect sensor from Lake\,Shore (model 460, Westerville, USA)~\cite{LakeShore2014} for the measurement of the static fields of the MPI scanner. Its accuracy is the sum of the reading error of about $\pm\SI{0.10}{\percent}$ and $\pm\SI{0.005}{\percent}$ of the chosen range \cite{LakeShore2014}. The used range of $\pm\SI{0.3}{\tesla}$ results in a maximum reading error of $\pm\SI{300}{\micro\tesla}$ and a range error of $\pm\SI{15}{\micro\tesla}$. The coil sensitivity of the dynamic drive field is measured with a three-axis coil sensor, which is connected to an \ac{ADC} for a digitization of the induced voltage signal~\cite{Thieben2022}. Each coil has a radius of \SI{2.5}{\mm} and an accuracy of about $\pm\SI{8}{\percent}$. The magnetic fields are measured in our preclinical MPI system 25/20FF (Bruker BioSpin MRI GmbH, Ettlingen, Germany), which is equipped with a three-axis Cartesian robot (isel Germany AG, Eichenzell, Germany) for an easy and accurate positioning of the measurement devices. The robot has a repetition accuracy in each direction of $\pm\SI{0.02}{\milli\metre}$ and an angle error of $\pm\SI{5}{\percent}$ for a motor step angle of \SI{1.8}{\degree}. Taking the accuracy of the gaussmeter or the coil sensor into account measurement errors due to mislocation of the robot are small and can be neglected. All fields are measured at the $36$ points of an $8$-design~\cite{Hardin8Design}, which are approached by the robot. In Fig.~\ref{fig:measurementSetup} an exemplary $3$-design with $6$ points is shown. According to the \SI{12}{\cm} diameter of the scanner bore, the points are rescaled to obtain a sphere with radius $\SI{42}{\milli\meter}$. The spheres were chosen as large as possible while keeping a safety margin that prevents collision of probe and scanner bore. The center $\brho$ is chosen near the \ac{FFP} of the selection field. At every point all three field directions are measured simultaneously.

For each of the magnetic fields we perform multiple measurements with different field strengths. Note, that the following field values describe the input parameters at the scanner, which ideally should result in the ideal magnetic fields described in the last section. First, the selection field at $10$ different gradient strengths of \SIrange{0.25}{2.5}{\tesla\per\m} with a step size of \SI{0.25}{\tesla\per\meter} is measured with the Hall sensor. The focus fields are measured with the Hall sensor as well with $11$ different shifts of \SIrange{-17}{17}{\milli\tesla} and a step size of \SI{3.4}{\milli\tesla} for the fields in $x$- and $y$-direction and \SIrange{-42}{42}{\milli\tesla} with a step size of \SI{8.4}{\milli\tesla} in $z$-direction. For the drive field measured with the coil sensor $4$ different amplitudes of \SIrange{6}{12}{\milli\tesla} with a step size of \SI{2}{\milli\tesla} for all three directions is set, as this range is most commonly used in our experiments. 

\begin{figure}[ht]
    \centering
    \includegraphics[width=0.8\textwidth]{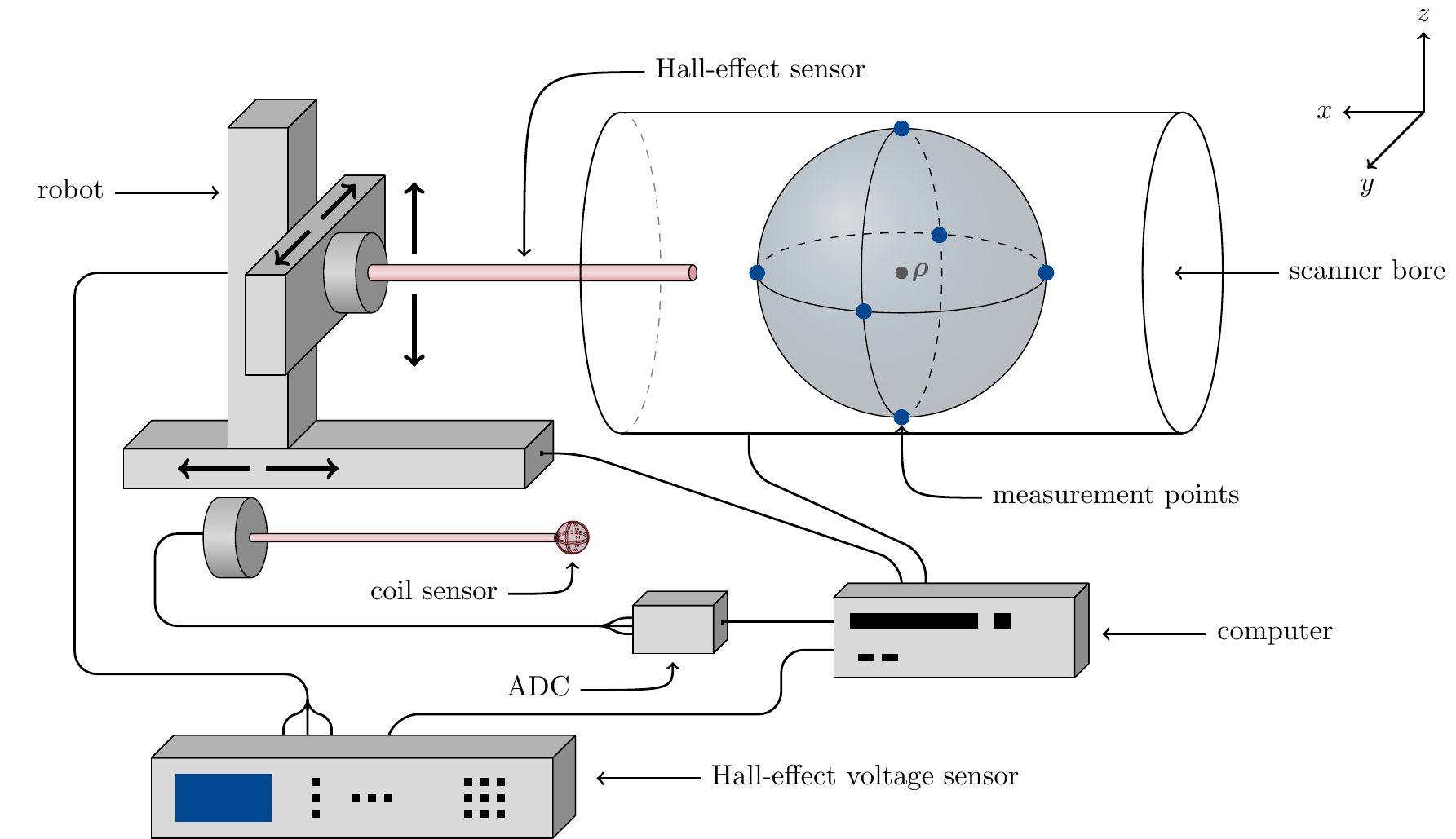}
    \caption{Measurement setup. The Hall-effect sensor of the gaussmeter is mounted on a three-axis Cartesian robot, which moves it to the chosen t-design positions inside the scanner bore. The voltage sensor of the gaussmeter transfers the measured data to the computer, which controls the robot movements and the settings of the \ac{MPI} scanner.}
    \label{fig:measurementSetup}
\end{figure}

%%%%%%%%%%%%%%%%%%%%%%%%%%
% Unique representation %%
%%%%%%%%%%%%%%%%%%%%%%%%%%
\subsection{Unique Representation}
While the single coils of the coil sensor are all centered around the same point, the three orthogonal detectors inside the Hall sensor are slightly shifted away from the center of the rod. Therefore, each detector measures the field on a slightly shifted sphere as it is shown in Fig.~\ref{fig:ShiftHallSensors}. All detectors are located \SI{1.8}{\mm} behind the tip of the rod and the $x$- and $y$-detector are additionally shifted by \SI{2.08}{\mm} outward from the center. This results in three expansions at slightly different expansion points $\brho_i$ for each direction. Using the translation from theorem~\ref{Th:Translation} the coefficients can be shifted into a common coordinate system centered at the tip of the rod $\brho$.

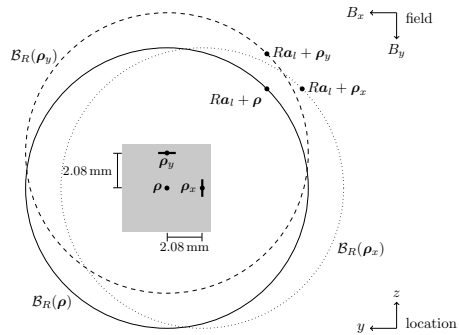
\begin{figure}
    \centering
    {\resizebox{6cm}{!}{\input{tikz/shiftGaussmeter.tex}}}
    \caption{The 3-axis Hall-effect sensor has three individual sensors for $x$-, $y$- and $z$-direction, respectively. The sensors for $x$- and $y$-direction ($\brho_x$ and $\brho_y$) are shifted off center inside the sensor rod (gray square). For each sensor the corresponding sphere $\ball_R(\brho_x)$ and $\ball_R(\brho_y)$ on which the magnetic field is measured are shown, as well as the sphere $\ball_R(\brho)$ at the tip of the rod. Additionally, the spatial coordinate system of the MPI scanner is shown on the bottom right. Above, the magnetic field coordinate system is displayed as it is given by the detector orientation of the sensor.}
    \label{fig:ShiftHallSensors}
\end{figure}

In MPI, the main selection field has a unique \ac{FFP}. This can be exploited by shifting the coefficients of the magnetic field expanded at the tip of the rod further into the \ac{FFP}, which we determine from the expansions using Newton's method. With this, the coefficients are independent from both, the measurement device itself and the entire measurement setup and do only depend on the MPI scanner specific fields. Using the \ac{FFP} of the selection field, we can also shift the coefficients of the drive and focus fields into this point. Altogether we obtain unique coefficients representing the magnetic fields of our MPI scanner.

%%%%%%%%%%%%%%%%%%%%
%% Implementation %%
%%%%%%%%%%%%%%%%%%%%
\subsection{Implementation}
All numerical methods described so far are implemented in the programming language Julia (version 1.8)~\cite{bezanson2017julia} in the open-source software package \texttt{SphericalHarmonicExpansions.jl} (version 0.1)~\cite{SphericalHarmonicExpansions}. The package provides methods for storage and handling of the coefficients of spherical or solid expansions, an efficient quadrature based on t-designs to calculate spherical or solid coefficients, a method to translate coefficients to a different expansion point and methods for fast numerical evaluation of the expansions in Cartesian coordinates. Furthermore, a collection of spherical t-designs can be obtained via the \texttt{MPIFiles.jl} package (version 0.12)~\cite{Knopp2019}. An example script, which shows how to obtain the expansion of a \SI{2}{\tesla\per\meter} selection field from the measurements described above is provided at \url{https://github.com/IBIResearch/SphericalHarmonicExpansionOfMagneticFields}.

%%%%%%%%%%%%%%%%%%%%%%%
%% Error Propagation %%
%%%%%%%%%%%%%%%%%%%%%%%
\subsection{Error Analysis}

The measurement instruments always have a tiny statistical independent measurement error $\delta > 0$, which is propagated to the coefficients and the final magnetic field using the laws of error propagation~\cite{Birge1939, Ferrero2006}. In our measurement setup there are two contributors to error. I.e. the uncertainty in the magnetic field measurements and errors in the positioning of the Hall-effect sensor. A systematic positional error may result from a non-ideal robot mount of the measurement sensors or bending of the rods to which those are attached. Such a systematic error effectively results in a shift of the coordinate system, which is compensated by the shift into the \ac{FFP} and hence can be neglected. Only non-systematic mislocations need to be accounted for, which are so small in our setup that they are neglected.

\begin{corollary}\label{Cor:ErrorProp}
    Let $\bgam^i(\brho)$ be the solid coefficients calculated with proposition~\ref{Prop:CoeffsTDesign} and $\bm B^\brho$ the resulting magnetic field calculated with proposition~\ref{Prop:MagnField_SphExp}. Additionally, we have the independent observational errors $\delta_k\!\left(R\bm a_k + \brho\right)$ of the measured boundary condition $\hat{B_i}{\left(R\bm a_k + \brho\right)}$.
    \renewcommand{\labelenumi}{\textit{\roman{enumi})}}
    \begin{enumerate}
        \item  The standard deviation for the propagated error of the coefficient $\gamma_{l,m}^i(\brho)$ can be obtained by
        \begin{align*}
            \epsilon_{l,m}^i(\brho) 
            = \frac{2l+1}{N R^l} \sqrt{\sum_{k=1}^N \left(\delta_k\!\left(R\bm a_k + \brho\right)Z_l^m(\bm a_k)\right)^2}.
        \end{align*}
        \item For each component of the magnetic field $B_i^\brho$, the standard deviation for the propagated error at a position $\bm a\in \ball_R(\brho)$ can be calculated by
        \begin{align*}
            \hat{\epsilon}_i^\brho(\bm a) 
            = \sqrt{ \sum_{k=1}^N 
            \left(\delta_k(R\bm a_k + \brho) \; \cS_L\!\left(\frac{2l+1}{NR^l} Z_l^m(R\bm a_k+\brho)\right)\!(\bm a)\right)^{\negthickspace 2}}.
        \end{align*}
    \end{enumerate}
\end{corollary}

\begin{remark}
    Propagating the error of the coefficients through the translation mapping can be done analogously since the translation mapping is linear in the coefficients.
\end{remark}

For error analysis we compare the field values provided by the truncated expansion to the measured ones. In our example study, we use the field measurements obtained at the spherical t-design positions which we also used to create the truncated expansion. In a typical application scenario, independent measurements should be used for error estimation. For the selection field, this is done by
\begin{align}\label{eq:MeasurementError}
    \zeta_i(k) = \frac{B_i^\brho(R \bm a_k) - \hat{B}_i(R\bm a_k + \brho)}{R g_i},
\end{align}
which compares the measured field $\hat{B}_i$ at a scaled t-design position $R\bm a_k + \brho$ with the calculated field $B_i^\brho$ at the same position. The difference is normalized to the scaled gradient strength $g_i$ of the considered direction $i$. The propagated error undergoes the same normalization $\hat{\zeta}_i(k) = \frac{\hat{\epsilon}_i^\brho(R\bm a_k+\brho)}{R g_i}$, which allows to assess the approximation quality of the truncated expansions.

%% file: tikz/shiftGaussmeter.tex
\begin{tikzpicture}[scale=0.42]

\fill[color=gray!42] (-2.5,-2.5) rectangle (2.5,2.5);
\draw (0,0)circle(8);
\fill (0,0)circle(4pt) node[left] {$\brho$};
\node at (-6.5,-6.5) {$\mathcal{B}_R(\brho)$};

% x
\draw[line width = 1.5pt] (2,-0.5) -- (2,0.5);
\fill (2,0)circle(4pt) node[left] {$\brho_x$};
\draw[dotted] (2,0)circle(8);
\node at (11,-3.5) {$\mathcal{B}_R(\brho_x)$};
% y
\draw[line width = 1.5pt] (-0.5,2) -- (0.5,2); 
\fill (0,2)circle(4pt) node[below] {$\brho_y$};
\draw[dashed] (0,2)circle(8);
\node at (-7.5,7.5) {$\mathcal{B}_R(\brho_y)$};

% Punkte
\fill (5.657,5.657)circle(4pt) node[below left] {$R\bm a_l+\brho$};
\fill (7.657,5.657)circle(4pt) node[right] {$R\bm a_l+\brho_x$};
\fill (5.657,7.657)circle(4pt) node[right] {$R\bm a_l+\brho_y$};

% Koordinatensystem (Feld):
\begin{scope}[xshift=13cm,yshift=10cm]
    \draw[->] (0,0) -- (0,-1.5) node[below] {$B_y$};
    \draw[->] (0,0) -- (-1.5,0) node[left] {$B_x$};
    \node[right,anchor=north west, inner sep=0pt] at (0.5,0) {field};
\end{scope}

% Koordinatensystem (Ort):
\begin{scope}[xshift=13cm,yshift=-8cm]
    \draw[->] (0,0) -- (0,1.5) node[above] {$z$};
    \draw[->] (0,0) -- (-1.5,0) node[left] {$y$};
    \node[right,anchor=south west, inner sep=0pt] at (0.5,0) {location};
\end{scope}

% Masse
\draw[|-|] (0,-2.8) --node[below] {\small \SI{2.08}{\mm}} (2,-2.8);
\draw[|-|] (-2.8,0) --node[left] {\small \SI{2.08}{\mm}} (-2.8,2);

\end{tikzpicture}

%% file: 4_Results.tex
\section{Results}\label{sec:Results}

%%%%%%%%%%%%%%%%%%%%%%%%
%% First Coefficients %%
%%%%%%%%%%%%%%%%%%%%%%%%

Exemplary, we examine the results of a \SI{2}{\tesla\per\meter} selection field. In the left part of Table~\ref{tab:initialCoeffs}, the initial coefficients calculated from the measurement without any post processing are listed, while in the right part, the processed coefficients are listed. Postprocessing consists of normalizing with the radius of the measured sphere, correcting the shifts of the Hall sensors and shifting into the \ac{FFP} calculated using the initial coefficients. Due to the shift into the FFP, the coefficients for $[l,m] = [0,0]$ are equal to zero up to floating point precision. Now, the gradient, which slightly deviates from the ideal gradient (cf. Table~\ref{tab:IdealCoeffs}), can be read directly from the coefficients for $l=1$. All other non-zero coefficients can be attributed to either measurement error or imperfections of the selection field.

\input{tikz/table_coeffs.tex}

%%%%%%%%%
%% MPI %%
%%%%%%%%%
\subsection{Magnetic fields in MPI}

%%%%%%%%%%%%%%%%%%%
%% Static Fields %%
%%%%%%%%%%%%%%%%%%%
\subsubsection{Static Fields}

Using translation of the coefficients enables comparison of the different magnetic fields applied in \ac{MPI}. In Fig.~\ref{fig:StaticFields}, the static magnetic fields of the central and the lower left patch of Fig.~\ref{fig:MPISetup} are shown. In the first row, the coefficients of the selection field with a gradient strength of \SI{2}{\tesla\per\meter} at its FFP $\bxi_\textup{c}$ and at the center point of the lower left patch $\bxi_\textup{s}$ are shown. While the coefficients at $\bxi_\textup{c}$ do not feature any significant imperfections, in the shifted coefficients some imperfections for $l \geq 1$ occur. Since $\bxi_\textup{s}$ is not the FFP of the selection field, $\bgam_{0,0}^\SF(\bxi_\textup{s},R_\textup{s})$ are nonzero and contain information about the offset field. Using an additional focus field with the same offset field but opposite sign, the offset field can be canceled and the FFP is shifted into $\bxi_\textup{s}$. This focus field is shown in the second row with an offset field of \SI{-24}{\milli\tesla} in $x$- and \SI{24}{\milli\tesla} in $z$-direction. At both positions some imperfections occur but they are slightly higher at the off-center position $\bxi_\textup{s}$. The combined selection and focus field is visualized in the last row. The axes of the field plot on the right are shifted due to the translation of the coefficients to the \ac{FFP} $\bxi_\textup{s}$. The FFP therefore has the coordinates $(0,0)$, as it is in the top field plot. Due to the shift into the FFP, the coefficients of the initial selection field and the combined field can be directly compared. It can be observed that the combined field has much more imperfections than the initial selection field, starting already from $l=1$.

\begin{figure}
    \centering
    \includegraphics[width=\textwidth]{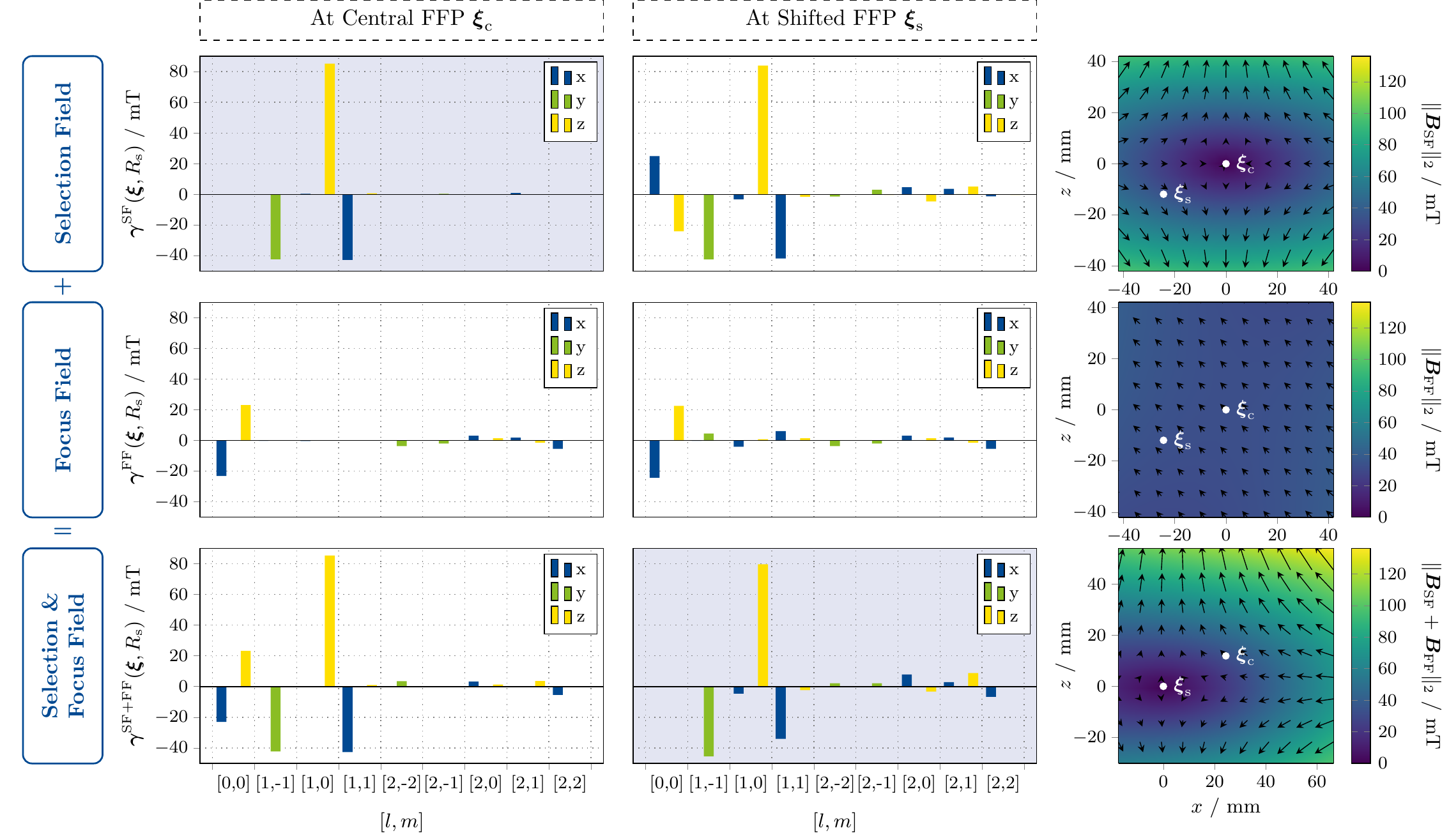}
    \caption{Solid harmonic analysis of the static fields in MPI, i.e. the selection and focus fields. The first row shows a selection field with a gradient strength of \SI{2}{\tesla\per\meter}. In the first column, the coefficients at the FFP of the selection field $\bxi_\textup{c}$ are shown, while in the second column the coefficients at another point $\bxi_\textup{s}$ are shown. Both points are marked in the field plot on the right.
    In the second row, a focus field of \SI{-24}{\milli\tesla} in $x$- and \SI{24}{\milli\tesla} in $z$-direction at both positions is shown. This additional field is required to shift the FFP from $\bxi_\textup{c}$ to $\bxi_\textup{s}$. The combined selection and focus field is shown in the last row. In an ideal MPI system the coefficients with light blue background would be identical.}
    \label{fig:StaticFields}
\end{figure}

The coefficients do not only enable comparison, they also allow for calculation of the real gradient strength and focus field shifts differing from the input parameters given in section~\ref{ssec:MeasurementSetup}. In case of the setup of Fig.~\ref{fig:StaticFields}, the real gradient strength of the selection field in its \ac{FFP} is \SIlist{-1.02;-1.01;2.03}{\tesla\per\meter} in $x$-, $y$-, and $z$-direction. Meanwhile, the real focus field shifts are \SIlist{-22.96;23.28}{\milli\tesla} in $x$- and $z$-direction, respectively.

%%%%%%%%%%%%%%%%%%%%
%% Dynamic Fields %%
%%%%%%%%%%%%%%%%%%%%
\subsubsection{Dynamic Fields}

The \SI{12}{\milli\tesla} dynamic fields of our MPI scanner are shown in Fig.~\ref{fig:DriveFieldsPatches}. As for the static fields, the coefficients at the FFP of the selection field $\bxi_\textup{c}$ and at the shifted FFP $\bxi_\textup{s}$ are shown in the left columns while the $x$-, $y$-, and $z$-drive fields are shown in the three rows. Overall, the coefficients decrease as $l$ increases, which justifies truncating the expansion at $L=4$. It can be observed that even in the center imperfections especially for $l > 1$ occur, which are more severe for the $y$- and $z$-drive field. The coefficients for $l=0$ show that the drive-field amplitudes \SIlist{12.35;12.29;12.15}{\milli\tesla} for the $x$-, $y$-, and $z$-drive field deviate from the \SI{12}{\milli\tesla} input parameter. This is because the drive-field coils are located closer to the scanner bore, than the selection- and focus-field coils. Comparing the coefficients at the central FFP and at the shifted FFP, the imperfections of the drive fields increase. Here, imperfections also arise for $l=0$, which are especially visible for the $z$-drive field. In the shifted FFP, the constant part of the $z$-drive field is not perfectly aligned in $z$-direction but also points slightly in the $x$-direction. In combination with the imperfections of the selection and focus field, this leads to the distorted trajectory of the lower left patch in Fig.~\ref{fig:MPISetup}. The imperfections also manifest in the field plot on the right, where for each drive field a representative plane is shown.

\begin{figure}
    \centering
    \includegraphics[width=\textwidth]{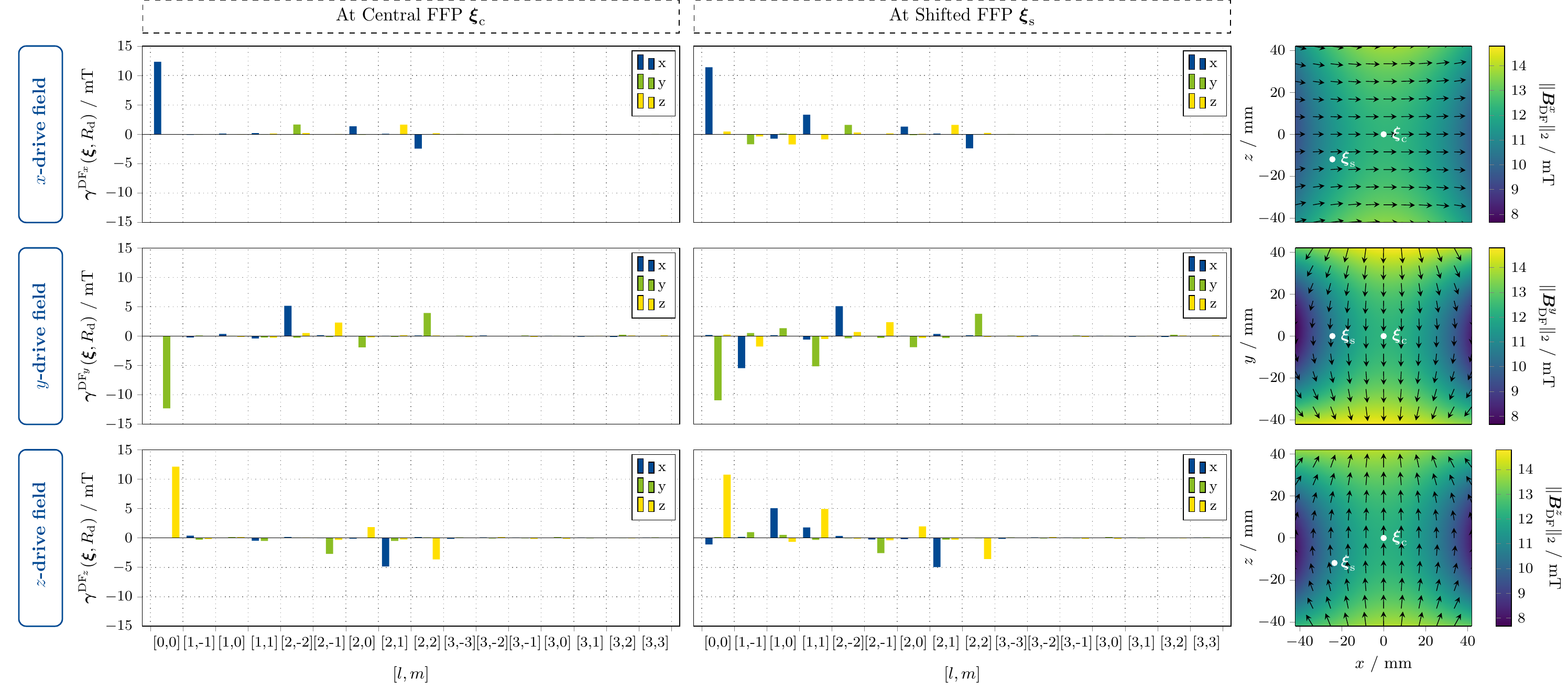}
    \caption{Solid harmonic analysis of the dynamic fields in MPI, i.e. the drive fields. Three drive fields in $x$-, $y$-, and $z$-direction with \SI{12}{\milli\tesla} amplitude are shown in each row. In the left columns, the coefficients up to $L=3$ at the FFPs of the selection fields of Fig.~\ref{fig:StaticFields} are visualized, while on the right, the fields in the $xz$- respectively $xy$-plane are shown.}
    \label{fig:DriveFieldsPatches}
\end{figure}

%%%%%%%%%%%%%%%%%%%%
%% Error analysis %%
%%%%%%%%%%%%%%%%%%%%
\subsection{Error Analysis}

A directional comparison of the field values provided by the truncated expansion to the measured ones at different gradient strengths shows an error (standard deviation) in the range of \SIrange{4}{100}{\micro\tesla}. As shown in Fig.~\ref{fig:Error}, this error increases with the gradient strength, is approximately the same in $y$- and $z$-direction, and approximately a factor of two smaller in $x$-direction. If we normalize the error as defined in eq.~\eqref{eq:MeasurementError}, one observes errors in the range of \numrange{4e-4}{17e-4}. The largest normalized errors can be observed for the smallest gradient strength, but no clear trend is evident for the remaining gradient strengths. Concerning the spatial dependence, the same observations apply which we just made. 

If we propagate the uncertainty of the calibration measurement, we have to expect errors in the range of \SIrange{5}{18}{\micro\tesla}, which increase with the gradient strength, as shown in Fig.~\ref{fig:Error}. These are similar in the $x$- and $y$-directions and stronger in the $z$-direction in the range of \SIrange{20}{70}{\percent} percent, increasing linearly with gradient strength. In direct comparison the observed error is up to a factor of \num{6} larger then the propagated one. Hence, the observed error can only partially attributed to uncertainties in the measurements with our Hall-effect sensor.

\begin{figure}
    \centering
    \includegraphics[width=0.7\textwidth]{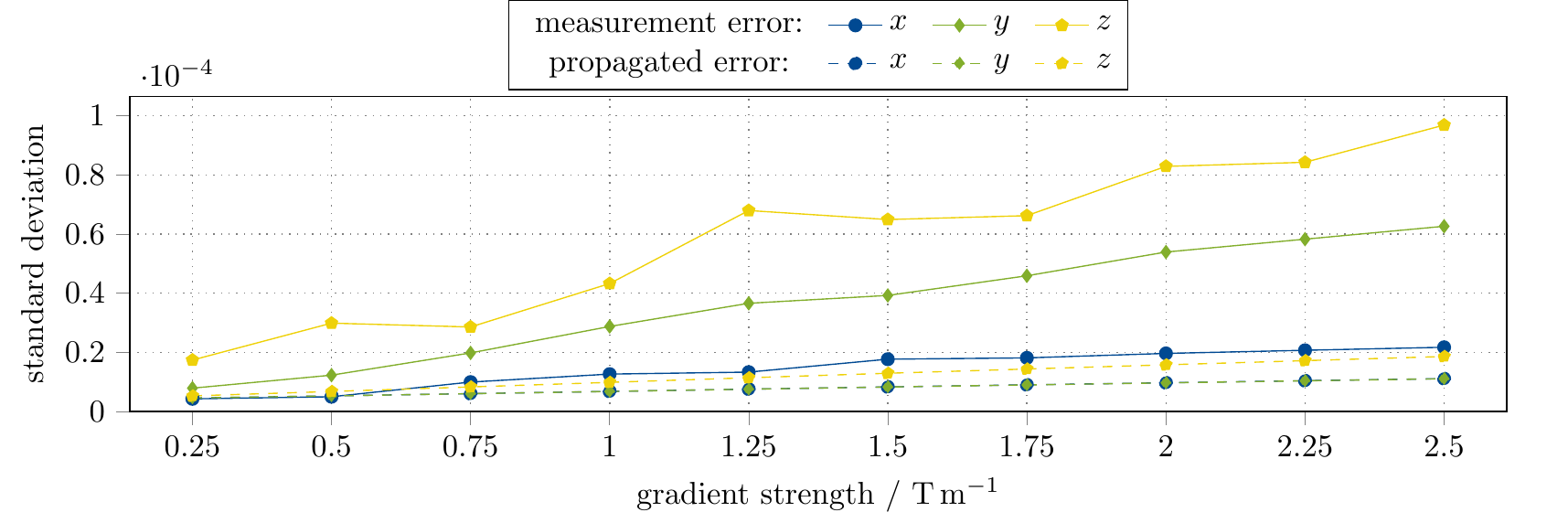}
    \caption{Standard deviation of the measurement errors (solid lines) and propagated errors (dashed lines).
    }
    \label{fig:Error}
\end{figure}

%% file: tikz/table_coeffs.tex
\renewcommand{\arraystretch}{1.2}
\pgfplotstableread[col sep=semicolon,]{tikz/Coeffs/data/CoeffsSF_initial-final.csv}\datatable
\begin{table}[ht]
    \centering
    \pgfplotstabletypeset[columns={num, xInitial, yInitial, zInitial,
                                  xFinal, yFinal, zFinal},
                          col sep=semicolon,
                          brackets/.style={%
                            postproc cell content/.append style={/pgfplots/table/@cell content/.add={\relax[}{]}},
                          },
                          every head row/.style={
                            before row={
                                \hline
                                &
                                \multicolumn{3}{c||}{\textbf{Initial Coefficients}} & \multicolumn{3}{c|}{\textbf{Processed Coefficients}}\\
                                },
                            after row=\hline\hline,},
                          every last row/.style={
                            after row=\hline},
                          every column/.style={
                            column type=|c},
                          columns/num/.style={
                            verb string type,
                            brackets,
                            column name={$[l,m]$},
                            column type=|c|},
                          columns/xInitial/.style={
                            column name=$x$},
                          columns/yInitial/.style={
                            column name=$y$},
                          columns/zInitial/.style={
                            column name=$z$,
                            column type=|c|},
                          columns/xFinal/.style={
                            column name=$x$},
                          columns/yFinal/.style={
                            column name=$y$},
                          columns/zFinal/.style={
                            column name=$z$,
                            column type=|c|},
                          my row iterator/.style={every row no #1/.style={
                            before row={
                              \hdashline
                              }
                            }
                          },
                          my row iterator/.list={1,4,9},
                          ]
        {\datatable}
    \caption{Comparison of the initial coefficients (in \si{\tesla}, left) and the normalized processed coefficients (in \si{\tesla\per\meter\tothe{\mathnormal{l}}}, right) of a \SI{2}{\tesla\per\meter} selection field.}
    \label{tab:initialCoeffs}
\end{table}
\renewcommand{\arraystretch}{1.0}

%% file: 5_Discussion.tex
\section{Discussion}

%% Introduction
In this paper we have given a review of the real solid harmonic expansions as a general solution of Laplace's equation and efficient quadrature methods for the calculation of the expansion coefficients via spherical t-designs. Furthermore, we proposed a method to change the reference point of the expansion using spatial shifts and thus arrive at a unique measurement setup independent expansion of the magnetic field. Moreover, we have shown how to analyze field imperfections using the polynomial structure of the expansion around the reference point. Our methods were evaluated on the signal generating and encoding fields in \acl{MPI}, where the coefficients provide a compact representation of the fields using the characteristic \acl{FFP} of the static selection field as unique expansion center. This uniqueness allows for comparison of the behavior of magnetic fields of different patches or \ac{MPI} scanner.

%% Measurement time
One of the main advantages of using these truncated expansions for the approximation of magnetic fields is the extremely fast acquisition time. Contrary to classical methods where the magnetic field is densely measured on the entire \ac{FOV}, fewer measurement on the surface of a sphere suffice to obtain the truncated expansion into real solid harmonics, regardless of the polynomial degree of the underlying field. In our setup measuring at the positions of a spherical $8$-design takes about \SI{2}{\minute}, which is sufficient to approximate the static and dynamic fields in MPI. In comparison, a Gauss-Legendre quadrature scheme, which has been commonly used in an MPI scenarios, has \SI{25}{\percent} more nodes and would take about \SI{2.5}{\minute}.

%% Accuracy
Our error analysis has shown that the deviations between model (truncated expansion) and measurement on the sphere are in the per mill range. However, only a small part of the observed deviations, \SI{20}{\percent} in the worst case scenario, can be attributed to the measurement inaccuracy of the Hall-effect sensor used. The error source with the greatest influence must therefore have a different origin. For example, a model error caused by the truncation of the expansion is possible. This hypothesis could be tested for example by choosing an expansion with larger $L$. For this, of course, a new spherical t-design with $t=2L$ would have to be chosen. However, for the application in the MPI context targeted in this work, the approximation accuracy of the real solid harmonic expansion up to degree $L=4$ achieved in this work is sufficient.

%% Field imperfections in MPI - Auswirkungen und Korrekturen
The coefficients allow for easy comparison of the different field setting of an \ac{MPI} system. Indeed, we observe slight imperfections in the shifted spatial encoding field and excitation field as shown in Fig.~\ref{fig:StaticFields} and Fig.~\ref{fig:DriveFieldsPatches}, respectively. These imperfections are a major cause for imaging artifacts and their precise knowledge is key in their reduction. First of all, the knowledge can be exploited for \ac{MPI} measurement planning. Shifting a patch to the correct position is crucial in many scenarios like multi-resolution data acquisition~\cite{gdaniec2017fast} or magnetic actuation~\cite{Rahmer2017}. Especially in multi-patch \ac{MPI}, the distorted shape and position of the shifted patches can lead to uncovered areas inside the \acl{FOV} as shown in Fig.~\ref{fig:MPISetup} or to imaging artifacts when the imperfections are not included in the patch-wise imaging operator~\cite{Szwargulski2018efficient}. The latter can be avoided by dedicated measurements of the operator of each patch for non-negligible field deviations~\cite{Boberg2020}, a field-dependent post-processing of the operator~\cite{Boberg2020IWMPI} or modeling of the imaging operator with integrated field imperfections~\cite{Albers2021}. Furthermore, spherical harmonic expansions can be directly incorporated in the reconstruction process~\cite{Bringout2020}. 

% cylindrical coordinates
Medical imaging setups often feature cylindrical gantries, so an expansion of the magnetic fields with cylindrical harmonics would be a natural choice. Nevertheless, an approximation with cylindrical harmonic expansions requires considerably more coefficients, which complicates the analysis of the magnetic field imperfections. Furthermore, we observed that even more coefficients are required to obtain a similar accuracy as with a spherical harmonic expansion since the basis functions of the cylindrical harmonic expansions are not as suitable for the presented magnetic fields. To the best of our knowledge, such a small set of quadrature nodes as the spherical t-design does not exist for quadrature on a cylindrical surface.

%% Future work
The compact representation of the magnetic fields in \ac{MPI} offer multiple further investigations. Using the presented tools we can deal with the field's imperfections in various applications. First of all, they are an important parameter for model based reconstructions. Incorporating the real field parameter into the modeled imaging operator lead to reconstruction results closer to those obtained with a measured operator. Thus, with the provided tools, we can work with the imperfections of the magnetic fields instead of avoiding them at all costs. They can be accounted for in the imaging sequences or for magnetic actuation. Furthermore, the spherical t-design offers sufficient small set of measurement points such that multiple Hall-effect sensors can be used simultaneously to measure a magnetic field in one shot. This can be used for direct feedback for magnetic field calibrations.

%% file: 6_Appendix1.tex
\subsection{Addition Theorem for Normalized Real Solid Harmonics}\label{ssec:Ap_AdditionHarmonics}

The translation of the solid harmonic coefficients is based on the addition theorem for normalized real solid harmonics, which is adapted from the addition theorem for unnormalized real solid harmonics presented in~\cite{Rico2013}. Let $z_l^m$ denote the unnormalized real solid spherical harmonics defined in~\cite{Rico2013}. The mapping of the normalized solid harmonics $Z_l^m$ used in this paper to the unnormalized ones is given by
\begin{align*}
    z_l^m = (-1)^m \sqrt{ \frac{ (l+\abs{m})! }{ (l-\abs{m})! } } Z_l^m 
        \begin{cases}
            \frac{1}{\sqrt{2}}, & m \neq 0\\
            1, & m = 0.
        \end{cases}
\end{align*}
Applied to equations~(6) and (7) in~\cite{Rico2013}, this yields for $m \geq 0$
\begin{equation}\label{eq:Ap_AddtionHarmonicsPos}
\begin{alignedat}{3}
    \tau_{\bm v}\big(Z_l^m(\bm a)\big)
    % 0 <= mu <= m
    =& \sum_{\lambda=0}^l\sum_{\mu=\max\set{0,\lambda-l+m}}^{\min\set{\lambda,m}} 
        &&\sigma^{(1)}_{l,m}(\lambda,\mu) 
        \left[Z_\lambda^\mu(\bm a) Z_{l-\lambda}^{m-\mu}(\bm v)
        - (1-\delta_{\mu 0})(1-\delta_{\mu m}) 
        Z_\lambda^{-\mu}(\bm a) Z_{l-\lambda}^{-(m-\mu)}(\bm v)\right]\\
    % mu > m
    +&\sum_{\lambda=m+1}^{l-1}\sum_{\mu=m+1}^{\min\set{\lambda,-\lambda+l+m}} 
        &&\sigma^{(2)}_{l,m}(\lambda,\mu)
        \left[Z_\lambda^\mu(\bm a) Z_{l-\lambda}^{\mu-m}(\bm v)
        + Z_{\lambda}^{-\mu}(\bm a) Z_{l-\lambda}^{-(\mu-m)}(\bm v)\right]\\
    % mu < 0
    +&\sum_{\lambda=1}^{l-m-1}\sum_{\mu=\max\set{-\lambda,\lambda-l+m}}^{-1} 
        \negmedspace&&\sigma^{(3)}_{l,m}(\lambda,\mu)
        \left[Z_\lambda^{-\mu}(\bm a) Z_{l-\lambda}^{m-\mu}(\bm v)
        + Z_\lambda^{\mu}(\bm a) Z_{l-\lambda}^{-(m-\mu)}(\bm v)\right]
\end{alignedat}
\end{equation}
and for $m < 0$
\begin{equation}\label{eq:Ap_AddtionHarmonicsNeg}
\begin{alignedat}{3}
    \tau_{\bm v}\big(Z_l^m(\bm a)\big) 
    % mu <= 0
    =&\sum_{\lambda=0}^l\sum_{\mu=\max\set{-\lambda,m}}^{\min\set{0,-\lambda+l+m}} 
        &&\sigma^{(1)}_{l,m}(\lambda,\mu)
        \negthinspace\left[(1\negthinspace-\negthinspace\delta_{\mu m})Z_\lambda^{-\mu}(\bm a) Z_{l-\lambda}^{-(\abs{m}+\mu)}(\bm v) +
        (1\negthinspace-\negthinspace\delta_{\mu 0})Z_\lambda^{\mu}(\bm a) Z_{l-\lambda}^{\abs{m}+\mu}(\bm v)\right]\\
    % mu < m (< 0)
    +&\sum_{\lambda=-m+1}^{l-1}\sum_{\mu=\max\set{-\lambda,\lambda-l+m}}^{m-1} 
        \hspace{-21pt}&&\sigma^{(2)}_{l,m}(\lambda,\mu)
        \negthinspace\left[-Z_\lambda^{-\mu}(\bm a) Z_{l-\lambda}^{\mu+\abs{m}}(\bm v) +
        Z_\lambda^{\mu}(\bm a) Z_{l-\lambda}^{-(\mu+\abs{m})}(\bm v)\right]\\
    % mu > 0
    +&\sum_{\lambda=1}^{l+m-1}\sum_{\mu=1}^{\min\set{\lambda,-\lambda+l+m}} 
        &&\sigma^{(3)}_{l,m}(\lambda,\mu)
        \negthinspace\left[Z_\lambda^\mu(\bm a) Z_{l-\lambda}^{-(\abs{m}+\mu)}(\bm v) -
        Z_\lambda^{-\mu}(\bm a) Z_{l-\lambda}^{\abs{m}+\mu}(\bm v)\right]
\end{alignedat}
\end{equation}

using the prefactors
\begin{align*}
    % factor for all factors
    \sigma_{l,m}(\lambda,\mu) &=
    \sqrt{\frac{(l+m)!(l-m)!}{(\lambda+\mu)!(\lambda-\mu)!(l-\lambda+m-\mu)!(l-\lambda-m+\mu)!}}\\
    % first sum
    \sigma^{(1)}_{l,m}(\lambda,\mu) &=
    \sigma_{l,m}(\lambda,\mu) \begin{cases}
        \frac{1}{\sqrt{2}}, &\mu \neq 0 \land \mu \neq m \land m \neq 0\\
        1, & \textup{else}
    \end{cases}\\
    % second sum
    \sigma^{(2)}_{l,m}(\lambda,\mu) &= \sigma_{l,m}(\lambda,\mu) \frac{(-1)^{\mu-m}}{2}
    \begin{cases}
        \sqrt{2}, & m \neq 0\\
        1, & m = 0
    \end{cases}\\
    % third sum
    \sigma^{(3)}_{l,m}(\lambda,\mu) &= \sigma^{(2)}_{l,m}(\lambda,\mu) (-1)^{m}
\end{align*}

and the Kronecker delta
\begin{align*}
    \delta_{ij} = 
    \begin{cases} 
        1, & i = j\\ 
        0, & \text{else.}
    \end{cases}
\end{align*}

%% file: 7_Appendix2.tex
\subsection{Addition Theorem Transferred to the Solid Harmonic Coefficients}\label{ssec:Ap_AdditionCoefficients}

\begin{proof}[Proof of theorem~\ref{Th:Translation}]
    As preparation to apply the addition theorem for the solid harmonics from the previous section, we split the sum into three parts where $m>0$, $m<0$, and $m=0$ holds:
    \begin{alignat}{3}
        \tau_{\bm v}\big(\cS_L(\bgam(\brho))(\bm a)\big) 
        = &\sum_{l=0}^L \sum_{m=-l}^l &&\gamma_{l,m}(\brho) \,\tau_{\bm v} \big(Z_l^m(\bm a)\big)\nonumber\\
        = &\sum_{l=1}^L \sum_{m=1}^l &&\gamma_{l,m}(\brho) \,\tau_{\bm v} \big(Z_l^m(\bm a)\big)\label{B:m>0}\\
        + &\sum_{l=0}^L &&\gamma_{l,0}(\brho) \,\tau_{\bm v} \big(Z_l^m(\bm a)\big)\label{B:m=0}\\
        + &\sum_{l=1}^L \sum_{m=-l}^{-1} &&\gamma_{l,m}(\brho) \,\tau_{\bm v} \big(Z_l^m(\bm a)\big).\label{B:m<0}
    \end{alignat}
    For each of the parts, the following three steps are applied. 
    \begin{enumerate}[label=\roman*)]
        \item Applying the addition theorem from section~\ref{ssec:Ap_AdditionHarmonics} to $\tau_{\bm v}\big(Z_l^m(\bm a)\big)$.
        \item Individual rearrangement of each of the terms into a form 
        \begin{align}\label{eq:StructureSummands}
            \sum_{l,m} Z_l^m(\bm a) \sum_{\lambda,\mu} \kappa(l,m,\lambda,\mu).
        \end{align}
        \item Setting $\hat{\tau}_{\bm v}\!\left(\gamma_{l,m}(\brho)\right) := \sum_{\lambda,\mu}\kappa(l,m,\lambda,\mu)$, which finally yields theorem~\ref{Th:Translation}.
    \end{enumerate}
    
    %%%%%%%%%%%%%%%%%%%
    %% First summand %%
    %%%%%%%%%%%%%%%%%%%
    \subsubsection*{Calculation of summand~\eqref{B:m>0}}
    In the first summand, it holds that $m>0$, which yields with \eqref{eq:Ap_AddtionHarmonicsPos}
    \begin{alignat}{3}
        \eqref{B:m>0} 
        = \sum_{l=1}^L \sum_{m=1}^l
        % sigma_1
        \Bigg[ &\sum_{\lambda=0}^l \sum_{\mu=\max\set{0,\lambda-l+m}}^{\min\set{\lambda,m}}
        \hspace{-10pt}&&\gamma_{l,m}(\brho)\,
        {\sigma}^{(1)}_{l,m}(\lambda,\mu) 
        \left(Z_{\lambda}^{\mu}(\bm a)Z_{l-\lambda}^{m-\mu}(\bm v)-(1-\delta_{\mu 0})(1-\delta_{\mu m})Z_{\lambda}^{-\mu}(\bm a)Z_{l-\lambda}^{-(m-\mu)}(\bm v)\right)\label{pos:1}\\
        % sigma_2
        +&\sum_{\lambda=m+1}^{l-1} \sum_{\mu=m+1}^{\min\set{\lambda,-\lambda+l+m}}
        \hspace{-12pt}&&\gamma_{l,m}(\brho)\,
        {\sigma}^{(2)}_{l,m}(\lambda,\mu) 
        \left(Z_{\lambda}^{\mu}(\bm a)Z_{l-\lambda}^{\mu-m}(\bm v)+Z_{\lambda}^{-\mu}(\bm a)Z_{l-\lambda}^{-(\mu-m)}(\bm v)\right)\label{pos:2}\\
        % sigma_3
        +&\sum_{\lambda=1}^{l-m-1} \sum_{\mu=\max\set{-\lambda,\lambda-l+m}}^{-1}
        \hspace{-25pt}&&\gamma_{l,m}(\brho)\,
        {\sigma}^{(3)}_{l,m}(\lambda,\mu) 
        \left(Z_{\lambda}^{-\mu}(\bm a)Z_{l-\lambda}^{m-\mu}(\bm v)+Z_{\lambda}^{\mu}(\bm a)Z_{l-\lambda}^{-(m-\mu)}(\bm v)\right)\Bigg].\label{pos:3}
    \end{alignat}
    Now, each summand is rearranged to obtain the structure from~\eqref{eq:StructureSummands}. For that purpose, the sums over $l$ and $\lambda$ and the sums over $m$ and $\mu$ have to be switched to factor out $Z_{\lambda}^{\mu}(\bm a)$. For the sake of simplicity, we omit the specific summand and indicate it with $\alpha$ and the important indices for the step.
    \begin{enumerate}
        \item We start with the transformation of summand~\eqref{pos:1}. The sums are swapped using the two reformulations
        \begin{align*}
            &\sum\limits_{l = 1}^L\sum\limits_{\lambda=0}^l \alpha_{l \lambda} 
            = \sum\limits_{\lambda = 1}^L\sum\limits_{l=\lambda}^L \alpha_{l \lambda} + \sum\limits_{l = 1}^L \alpha_{l 0},\\
            &\sum_{m = 1}^{l} \sum_{\mu=\max\set{0,\lambda-l+m}}^{\min\set{\lambda,m}} \alpha_{m \mu} 
            = \sum_{\mu=0}^\lambda \sum_{m = \max\set{1,\mu}}^{\mu-(\lambda-l)} \alpha_{m \mu}.
        \end{align*}
        Note that the second reformulation holds since $m>0$ in summand~\eqref{B:m>0}.
        Applying this to \eqref{pos:1} yields
        \begin{alignat*}{3}
            \eqref{pos:1} 
            = &&\sum_{\lambda=1}^L \sum_{\mu=1}^\lambda 
            &Z_\lambda^\mu(\bm a) \sum_{l=\lambda}^L \sum_{m = \mu}^{\mu-(\lambda-l)} \gamma_{l,m}(\brho)~ {\sigma}^{(1)}_{l,m}(\lambda,\mu)~Z_{l-\lambda}^{m-\mu}(\bm v)\\
            - &&\sum_{\lambda=1}^L \sum_{\mu=-\lambda}^{-1} 
            &Z_\lambda^\mu(\bm a) \sum_{l=\lambda}^L \sum_{m = -\mu+1}^{-\mu-(\lambda-l)} \gamma_{l,m}(\brho)~ {\sigma}^{(1)}_{l,m}(\lambda,-\mu)~Z_{l-\lambda}^{-(m+\mu)}(\bm v)\\
            && + \sum_{\lambda=1}^L 
            &Z_\lambda^0(\bm a)\sum_{l=\lambda}^L \sum_{m = 1}^{-(\lambda-l)} \gamma_{l,m}(\brho)~ {\sigma}^{(1)}_{l,m}(\lambda,0)~Z_{l-\lambda}^{m}(\bm v)\\
            && + &Z_0^0(\bm a) \sum_{l=1}^L \sum_{m = 1}^{l} \gamma_{l,m}(\brho)~ {\sigma}^{(1)}_{l,m}(0,0)~Z_{l}^{m}(\bm v),
        \end{alignat*}
        which has the structure of~\eqref{eq:StructureSummands} by relabeling the indices $\lambda$ with $l$ and $\mu$ with $m$. Note that the sign of $\mu$ is switched in the second summand in order to factor out $Z_{\lambda}^{\mu}$.
    
        \item Next, summand~\eqref{pos:2} is transformed by swapping the sums as
        \begin{align*}
            &\sum_{m=1}^{l}\sum_{\lambda=m+1}^{l-1} \alpha_{m \lambda}
            = \sum_{\lambda=1}^{l-1}\sum_{m=1}^{\lambda-1} \alpha_{m \lambda},\\
            &\sum_{l = 1}^L \sum_{\lambda=1}^{l-1} \alpha_{l \lambda}
            = \sum_{\lambda = 1}^L \sum_{l=\lambda+1}^{L} \alpha_{l \lambda},\\
            &\sum_{m=1}^{\lambda-1}\sum_{\mu=m+1}^{\min\set{\lambda,-\lambda+l+m}} \alpha_{m \mu}
            = \sum_{\mu=1}^\lambda \sum_{m=\max\set{1,\mu-(l-\lambda)}}^{\mu-1} \alpha_{m \mu}.
        \end{align*}
        Applying the swapped sums leads to
        \begin{alignat*}{2}
            \eqref{pos:2} = 
            &\sum_{\lambda=1}^L\sum_{\mu=1}^\lambda 
            && Z_\lambda^\mu(\bm a) \sum_{l=\lambda+1}^L\sum_{m=\max\set{1,\mu-(l-\lambda)}}^{\mu-1} \gamma_{l,m}(\brho)~ {\sigma}^{(2)}_{l,m}(\lambda,\mu)~Z_{l-\lambda}^{\mu-m}(\bm v)\\
            + &\sum_{\lambda=1}^L\sum_{\mu=-\lambda}^{-1} 
            &&Z_\lambda^\mu(\bm a) \sum_{l=\lambda+1}^L\sum_{m=\max\set{1,-\mu-(l-\lambda)}}^{-\mu-1} \gamma_{l,m}(\brho)~ {\sigma}^{(2)}_{l,m}(\lambda,-\mu)~Z_{l-\lambda}^{\mu+m}(\bm v),
        \end{alignat*}
        which again has the structure of~\eqref{eq:StructureSummands} by relabeling the indices $\lambda$ with $l$ and $\mu$ with $m$.
    
        \item Finally, summand~\eqref{pos:3} is transformed analogously. Switching the sums over $m$ and $\lambda$ by
        \begin{align*}
            \sum_{m=1}^{l}\sum_{\lambda=1}^{l-m-1} \alpha_{m \lambda}
            = \sum_{\lambda=1}^{l-1}\sum_{m=1}^{-\lambda+l-1} \alpha_{m \lambda},
        \end{align*}
        over $l$ and $\lambda$ as it is done for~\eqref{pos:2}, and over $m$ and $\mu$ by
        \begin{align*}
            \sum_{m=1}^{-\lambda+l-1}\sum_{\mu=\max\set{-\lambda,\lambda-l+m}}^{-1} \alpha_{m \mu}
            = \sum_{\mu=-\lambda}^{-1}\sum_{m=1}^{\mu-(\lambda-l)} \alpha_{m \mu}
        \end{align*}
        leads to the transformation
        \begin{equation*}\begin{alignedat}{3}
            \eqref{pos:3} = 
            &\sum_{\lambda=1}^L \sum_{\mu=1}^\lambda& &Z_l^\mu(\bm a)\sum_{l=\lambda+1}^L \sum_{m=1}^{-\mu-(\lambda-l)} \gamma_{l,m}(\brho)~ {\sigma}^{(3)}_{l,m}(\lambda,-\mu)~Z_{l-\lambda}^{\mu+m}(\bm v)\\
            + &\sum_{\lambda=1}^L \sum_{\mu=-\lambda}^{-1}& &Z_l^\mu(\bm a)\sum_{l=\lambda+1}^L \sum_{m=1}^{\mu-(\lambda-l)} \gamma_{l,m}(\brho)~ {\sigma}^{(3)}_{l,m}(\lambda,\mu)~Z_{l-\lambda}^{-(m-\mu)}(\bm v).
        \end{alignedat}
        \end{equation*}
    \end{enumerate}
    
    Altogether, with an relabeling of the indices $l$ with $\lambda$ and $m$ with $\mu$ this leads to 
    \begin{alignat*}{3}
        \eqref{B:m>0} &&= \eqref{pos:1}+\eqref{pos:2}+\eqref{pos:3}\\
        %% m > 0
            % (1)
            &&= \sum_{l=1}^L\sum_{m=1}^l Z_l^m(\bm a)& 
            \Bigg[\sum_{\lambda=l}^L\sum_{\mu=m}^{m-(l-\lambda)} \gamma_{\lambda,\mu}(\brho)~ 
            {\sigma}^{(1)}_{\lambda,\mu}(l,m)~Z_{\lambda-l}^{\mu-m}(\bm v)\\
            % (2)
            && &+ \sum_{\lambda=l+1}^L\sum_{\mu=\max\set{1,m-(\lambda-l)}}^{m-1} \gamma_{\lambda,\mu}(\brho)~ 
            {\sigma}^{(2)}_{\lambda,\mu}(l,m)~Z_{\lambda-l}^{m-\mu}(\bm v)\\
            % (3)
            && &+ \sum_{\lambda=l+1}^L\sum_{\mu=1}^{-m-(l-\lambda)} \gamma_{\lambda,\mu}(\brho)~ 
            {\sigma}^{(3)}_{\lambda,\mu}(l,-m)~Z_{\lambda-l}^{m+\mu}(\bm v)\Bigg]\\
        %% m < 0
            % (1)
            &&+ \sum_{l=1}^L\sum_{m=-l}^{-1} Z_l^m(\bm a)& 
            \Bigg[-\sum_{\lambda=l}^L\sum_{\mu=-m+1}^{-m-(l-\lambda)} \gamma_{\lambda,\mu}(\brho)~
            {\sigma}^{(1)}_{\lambda,\mu}(l,-m)~Z_{\lambda-l}^{-(\mu+m)}(\bm v)\\
            % (2)
            && &+ \sum_{\lambda=l+1}^L\sum_{\mu=\max\set{1,-m-(\lambda-l)}}^{-m-1} \gamma_{\lambda,\mu}(\brho)~
            {\sigma}^{(2)}_{\lambda,\mu}(l,-m)~Z_{\lambda-l}^{m+\mu}(\bm v)\\
            % (3)
            && &+ \sum_{\lambda=l+1}^L\sum_{\mu=1}^{m-(l-\lambda)}\gamma_{\lambda,\mu}(\brho)~ 
            {\sigma}^{(3)}_{\lambda,\mu}(l,m)~Z_{\lambda-l}^{-(\mu-m)}(\bm v)\Bigg]\\
        %% m = 0
            && +\sum_{l=1}^L Z_l^0(\bm a)& \Bigg[\sum_{\lambda=l}^L\sum_{\mu=1}^{-(l-\lambda)} \gamma_{\lambda,\mu}(\brho)~
            {\sigma}^{(1)}_{\lambda,\mu}(l,0)~Z_{\lambda-l}^{\mu}(\bm v)\Bigg]\\
            % l = m = 0
            && +Z_0^0(\bm a)& \Bigg[\sum_{\lambda=1}^L\sum_{\mu=1}^{\lambda} \gamma_{\lambda,\mu}(\brho)~
            {\sigma}^{(1)}_{\lambda,\mu}(0,0)~Z_{\lambda}^{\mu}(\bm v)\Bigg].
    \end{alignat*}
    The parts contained in the square brackets are the first summands that define $\hat{\tau}_{\bm v}\gamma_{l,m}$ later on. In the following, the summands \eqref{B:m=0} and \eqref{B:m<0} are transformed analogously. 
    
    %%%%%%%%%%%%%%%%%%%%
    %% Second summand %%
    %%%%%%%%%%%%%%%%%%%%
    \subsubsection*{Calculation of summand~\eqref{B:m=0}}
    Since $m=0$ holds for the second summand, applying \eqref{eq:Ap_AddtionHarmonicsPos} leads to
    \begin{align}
        \eqref{B:m=0} =
        \sum_{l=0}^L
        % sigma 1
        \Bigg[&\sum_{\lambda=0}^l \gamma_{l,0}(\brho)~ \sigma^{(1)}_{l,0}(\lambda,0)~Z_{\lambda}^0(\bm a)Z_{l-\lambda}^0(\bm v)\label{null:1}\\
        % sigma 2
        +&\sum_{\lambda=1}^{l-1}\sum_{\mu=1}^{\min\set{\lambda,-\lambda+l}} \gamma_{l,0}(\brho)~
        \sigma^{(2)}_{l,0}(\lambda,\mu)~ \left(Z_{\lambda}^{\mu}(\bm a)Z_{l-\lambda}^{\mu}(\bm v)+Z_{\lambda}^{-\mu}(\bm a)Z_{l-\lambda}^{-\mu}(\bm v)\right)\label{null:2}\\
        % sigma 3
        +&\sum_{\lambda=1}^{l-1}\sum_{\mu=\max\set{-\lambda,\lambda-l}}^{-1} \gamma_{l,0}(\brho)~ 
        \sigma^{(3)}_{l,0}(\lambda,\mu)~\left(Z_{\lambda}^{-\mu}(\bm a)Z_{l-\lambda}^{-\mu}(\bm v)+Z_{\lambda}^{\mu}(\bm a)Z_{l-\lambda}^{\mu}(\bm v)\right)\Bigg].\label{null:3}
    \end{align}
    
    Again, the sums over $l$ and $\lambda$ and $m$ and $\mu$ are swapped to obtain the structure from~\eqref{eq:StructureSummands}. Since it is straightforward for~\eqref{null:1}, we directly start with summand~\eqref{null:2}.
    \begin{enumerate}
        \item Switching the sums is done by
        \begin{align*}
            &\sum_{l=0}^L\sum_{\lambda=1}^{l-1} \alpha_{l \lambda} 
            = \sum_{\lambda=1}^L\sum_{l=\lambda+1}^L \alpha_{l \lambda},\\
            &\sum_{l=\lambda+1}^L\sum_{\mu=1}^{\min\set{\lambda,-\lambda+l}} \alpha_{l \mu}
            = \sum_{\mu=1}^\lambda\sum_{l=\lambda+\mu}^L \alpha_{l \mu},
        \end{align*}
        which yields
        \begin{align*}
            \eqref{null:2} = 
            \sum_{\lambda=1}^L\sum_{\mu=1}^\lambda &Z_\lambda^\mu(\bm a) \sum_{l=\lambda+\mu}^L 
            \gamma_{l,0}(\brho)~ {\sigma}^{(2)}_{l,0}(\lambda,\mu)~Z_{l-\lambda}^\mu(\bm v)\\
            +\sum_{\lambda=1}^L\sum_{\mu=-\lambda}^{-1} &Z_\lambda^\mu(\bm a) \sum_{l=\lambda-\mu}^L \gamma_{l,0}(\brho)~ 
            {\sigma}^{(2)}_{l,0}(\lambda,-\mu)~Z_{l-\lambda}^\mu(\bm v).
        \end{align*}
    
        \item For the third summand~\eqref{null:3} we use the same transformation for $l$ and $\lambda$ and swap $l$ and $\mu$ by
        \begin{align*}
            \sum_{l=\lambda+1}^L\sum_{\mu=\max\set{-\lambda,\lambda-l}}^{-1} \alpha_{l \mu}
            = \sum_{\mu=-\lambda}^{-1}\sum_{l=\lambda-\mu}^L \alpha_{l \mu}.
        \end{align*}
        Combining the transformations, the third summand~\eqref{null:3} can be reformulated as
        \begin{align*}
            \eqref{null:3} = 
            \sum_{\lambda=1}^L\sum_{\mu=1}^\lambda &Z_\lambda^\mu(\bm a) \sum_{l=\lambda+\mu}^L \gamma_{l,0}(\brho)~ 
            {\sigma}^{(3)}_{l,0}(\lambda,-\mu)~Z_{l-\lambda}^\mu(\bm v)\\
            +\sum_{\lambda=1}^L\sum_{\mu=-\lambda}^{-1} &Z_\lambda^\mu(\bm a) \sum_{l=\lambda-\mu}^L \gamma_{l,0}(\brho)~ 
            {\sigma}^{(3)}_{l,0}(\lambda,\mu)~Z_{l-\lambda}^\mu(\bm v).
        \end{align*}
    \end{enumerate}
    
    Altogether by relabeling $l$ with $\lambda$ and $m$ with $\mu$, we get 
    \begin{alignat*}{3}
        \eqref{B:m=0}~ &&= \eqref{null:1} + \eqref{null:2} + \eqref{null:3}\\
        % m > 0
            % (2)
            &&= \sum_{l=1}^L \sum_{m=1}^{l} Z_l^m(\bm a) 
            &\Bigg[ \sum_{\lambda=l+m}^L \gamma_{\lambda,0}(\brho)~ 
            {\sigma}^{(2)}_{\lambda,0}(l,m)~Z_{\lambda-l}^m(\bm v)\\
            % (3)
            && &+ \sum_{\lambda=l+m}^L \gamma_{\lambda,0}(\brho)~ 
            {\sigma}^{(3)}_{\lambda,0}(l,-m)~Z_{\lambda-l}^m(\bm v)\Bigg]\\
        % m < 0
            % (2)
            &&+ \sum_{l=1}^L \sum_{m=-l}^{-1} Z_l^m(\bm a) 
            &\Bigg[ \sum_{\lambda=l-m}^L \gamma_{\lambda,0}(\brho)~ 
            {\sigma}^{(2)}_{\lambda,0}(l,-m)~Z_{\lambda-l}^m(\bm v)\\
            % (3)
            && &+ \sum_{\lambda=l-m}^L \gamma_{\lambda,0}(\brho)~ 
            {\sigma}^{(3)}_{\lambda,0}(l,m)~Z_{\lambda-l}^m(\bm v)\Bigg]\\
        % m=0
        &&+ \sum_{l=1}^L Z_l^0(\bm a) & \Bigg[\sum_{\lambda=l}^L \gamma_{\lambda,0}(\brho)~  {\sigma}^{(1)}_{\lambda,0}(l,0)~Z_{\lambda-l}^0(\bm v)\Bigg]\\
        && + Z_0^0(\bm a) & \Bigg[\sum_{\lambda=0}^L ~\gamma_{\lambda,0}(\brho)~ 
        {\sigma}^{(1)}_{\lambda,0}(0,0)~Z_{\lambda}^0(\bm v)\Bigg].
    \end{alignat*}
    
    %%%%%%%%%%%%%%%%%%%
    %% Third summand %%
    %%%%%%%%%%%%%%%%%%%
    \subsubsection*{Calculation of summand~\eqref{B:m<0}}
    Finally, \eqref{eq:Ap_AddtionHarmonicsNeg} is applied to the third summand where $m<0$ holds, which yields
    \begin{align}
        \eqref{B:m<0} = 
        % sigma 1
        \sum_{l=1}^L \sum_{m=-l}^{-1} 
        \Bigg[ &\sum_{\lambda=0}^l\sum_{\mu=\max\set{-\lambda,m}}^{\min\set{0,-\lambda+l+m}}
        \hspace{-20.6pt} \gamma_{l,m}(\brho)~ \sigma^{(1)}_{l,m}(\lambda,\mu)
        \negthinspace\left((1\negthinspace-\negthinspace\delta_{\mu 0})Z_{\lambda}^{\mu}(\bm a)Z_{l-\lambda}^{\abs{m}+\mu}(\bm v) 
        \negthinspace+\negthinspace (1\negthinspace-\negthinspace\delta_{\mu m})Z_{\lambda}^{-\mu}(\bm a)Z_{l-\lambda}^{-(\abs{m}+\mu)}(\bm v)\right)\label{neg:1}\\
        % sigma 2
        + &\sum_{\lambda=-m+1}^{l-1}\sum_{\mu = \max\set{-\lambda,\lambda-l+m}}^{m-1}
        \gamma_{l,m}(\brho)~ \sigma^{(2)}_{l,m}(\lambda,\mu)~\left(Z_{\lambda}^{\mu}(\bm a)Z_{l-\lambda}^{-(\mu+\abs{m})}(\bm v) - Z_{\lambda}^{-\mu}(\bm a)Z_{l-\lambda}^{\mu+\abs{m}}(\bm v)\right)\label{neg:2}\\
        % sigma 3
        + &\sum_{\lambda=1}^{l+m-1}\sum_{\mu=1}^{\min\set{\lambda,-\lambda+l+m}} 
        \gamma_{l,m}(\brho)~ \sigma^{(3)}_{l,m}(\lambda,\mu)~\left(Z_{\lambda}^{\mu}(\bm a)Z_{l-\lambda}^{-(\abs{m}+\mu)}(\bm v) - Z_{\lambda}^{-\mu}(\bm a)Z_{l-\lambda}^{\abs{m}+\mu}(\bm v)\right)\Bigg].\label{neg:3}
    \end{align}
    
    Now, each summand is transformed analogously to \eqref{B:m>0}. 
    \begin{enumerate}
        \item For the first summand~\eqref{neg:1} the sums over $l$ and $\lambda$ are swapped as it is done for \eqref{pos:1}. Together with
        \begin{align*}
            \sum_{m=-l}^{-1}\sum_{\mu=\max\set{-\lambda,m}}^{\min\set{0,-\lambda+l+m}} \alpha_{m \mu}
            = \sum_{\mu=-\lambda}^0\sum_{m=\mu-(l-\lambda)}^{\min\set{-1,\mu}} \alpha_{m \mu}
        \end{align*}
        this yields
        \begin{align*}
            \eqref{neg:1} = 
            \sum_{\lambda=1}^L\sum_{\mu=-\lambda}^{-1} &Z_\lambda^\mu(\bm a) \sum_{l=\lambda}^L\sum_{m=\mu-(l-\lambda)}^{\mu} \gamma_{l,m}(\brho)~  
            {\sigma}^{(1)}_{l,m}(\lambda,\mu)~Z_{l-\lambda}^{\abs{m}+\mu}(\bm v)\\
            + \sum_{\lambda=1}^L\sum_{\mu=1}^{\lambda} &Z_\lambda^\mu(\bm a) \sum_{l=\lambda}^L\sum_{m=-\mu-(l-\lambda)}^{-\mu-1} \gamma_{l,m}(\brho)~ 
            {\sigma}^{(1)}_{l,m}(\lambda,-\mu)~Z_{l-\lambda}^{-(\abs{m}-\mu)}(\bm v)\\
            + \sum_{\lambda=1}^L &Z_\lambda^0(\bm a) \sum_{l=\lambda}^L\sum_{m=-(l-\lambda)}^{-1} \gamma_{l,m}(\brho)~ 
            {\sigma}^{(1)}_{l,m}(\lambda,0)~Z_{l-\lambda}^{-\abs{m}}(\bm v)\\
            + &Z_0^0(\bm a) \sum_{l=1}^L\sum_{m=-l}^{-1} \gamma_{l,m}(\brho)~ 
            {\sigma}^{(1)}_{l,m}(0,0)~Z_{l}^{-\abs{m}}(\bm v).
        \end{align*}
    
        \item For the second summand~\eqref{neg:2}, the sums are swapped by
        \begin{align*}
            &\sum_{m=-l}^{-1}\sum_{\lambda=-m+1}^{l-1} \alpha_{m \lambda} 
            = \sum_{\lambda=1}^{l-1}\sum_{m=-\lambda+1}^{-1} \alpha_{m \lambda},\\
            &\sum_{m=-\lambda+1}^{-1}\sum_{\mu=\max\set{-\lambda,m}}^{m-1} \alpha_{m \mu} 
            = \sum_{\mu=-\lambda}^{-1}\sum_{m=\mu+1}^{\min\set{-1,\mu-(\lambda-l)}} \alpha_{m \mu},
        \end{align*}
        and the sums over $l$ and $\lambda$ are swapped in the same way as done for \eqref{pos:2}.
        Applying this to \eqref{neg:2} yields
        \begin{align*}
            \eqref{neg:2} = 
            \sum_{\lambda=1}^L\sum_{\mu=-\lambda}^{-1}&Z_\lambda^\mu(\bm a)\sum_{l=\lambda+1}^L\sum_{m=\mu+1}^{\min\set{-1,\mu-(\lambda-l)}}
            \gamma_{l,m}(\brho)~ 
            {\sigma}^{(2)}_{l,m}(\lambda,\mu)~Z_{l-\lambda}^{-(\mu+\abs{m})}(\bm v)\\
            - \sum_{\lambda=1}^L\sum_{\mu=1}^{\lambda}&Z_\lambda^\mu(\bm a)\sum_{l=\lambda+1}^L\sum_{m=-\mu+1}^{\min\set{-1,-\mu-(\lambda-l)}}
            \gamma_{l,m}(\brho)~  
            {\sigma}^{(2)}_{l,m}(\lambda,-\mu)~Z_{l-\lambda}^{\abs{m}-\mu}(\bm v).
        \end{align*}
    
        \item Finally, summand~\eqref{neg:3} is transformed. Using the transformations
        \begin{align*}
            &\sum_{m=-l}^{-1}\sum_{\lambda=1}^{l+m-1} \alpha_{m \lambda}
            = \sum_{\lambda=1}^{l-1}\sum_{m=\lambda-l+1}^{-1} \alpha_{m \lambda},\\
            &\sum_{m=\lambda-l+1}^{-1}\sum_{\mu=1}^{\min\set{\lambda,-\lambda+l+m}} \alpha_{m \mu}
            = \sum_{\mu=1}^\lambda\sum_{m=\mu-(l-\lambda)}^{-1} \alpha_{m \mu},
        \end{align*}
        and for $l$ and $\lambda$ the transformation as it was done for \eqref{pos:2}, yields
        \begin{alignat*}{3}
            \eqref{neg:3} = 
            &&\sum_{\lambda=1}^L\sum_{\mu=1}^\lambda &Z_\lambda^\mu(\bm a) \sum_{l=\lambda+1}^L\sum_{m=\mu-(l-\lambda)}^{-1} \gamma_{l,m}(\brho)~ 
            {\sigma}^{(3)}_{l,\abs{m}}(\lambda,\mu)~Z_{l-\lambda}^{-(\abs{m}+\mu)}(\bm v)\\
            - &&\sum_{\lambda=1}^L\sum_{\mu=-\lambda}^{-1} &Z_\lambda^\mu(\bm a) \sum_{l=\lambda+1}^L\sum_{m=-\mu-(l-\lambda)}^{-1} \gamma_{l,m}(\brho)~ 
            {\sigma}^{(3)}_{l,\abs{m}}(\lambda,-\mu)~Z_{l-\lambda}^{ \abs{m}-\mu}(\bm v).
        \end{alignat*}
    \end{enumerate}
    
    Finally by relabeling $l$ with $\lambda$ and $m$ with $\mu$ the third summand~\eqref{B:m<0} now reads
    \begin{align*}
        \eqref{B:m<0} = \eqref{neg:1} + \eqref{neg:2} + \eqref{neg:3}\\
        % summand 1
            % (1)
            = \sum_{l=1}^L\sum_{m=1}^l Z_l^m(\bm a)& 
            \Bigg[\sum_{\lambda=l}^L\sum_{\mu=-m-(\lambda-l)}^{-m-1} \gamma_{\lambda,\mu}(\brho)~ 
            {\sigma}^{(1)}_{\lambda,\mu}(l,-m)~Z_{\lambda-l}^{-(\abs{\mu}-m)}(\bm v)\\
            % (2)
            &-\sum_{\lambda=l+1}^L\sum_{\mu=-m+1}^{\min\set{-1,-m-(l-\lambda)}}\gamma_{\lambda,\mu}(\brho)~ 
            {\sigma}^{(2)}_{\lambda,\mu}(l,-m)~Z_{\lambda-l}^{\abs{\mu}-m}(\bm v)\\
            % (3)
            &+\sum_{\lambda=l+1}^L\sum_{\mu=m-(\lambda-l)}^{-1} \gamma_{\lambda,\mu}(\brho)~ 
            {\sigma}^{(3)}_{\lambda,\mu}(l,m)~Z_{\lambda-l}^{-(\abs{\mu}+m)}(\bm v)\Bigg]\\
        % summand 2
            % (1)
            +\sum_{l=1}^L\sum_{m=-l}^{-1} Z_l^m(\bm a)& 
            \Bigg[\sum_{\lambda=l}^L\sum_{\mu=m-(\lambda-l)}^{m} \gamma_{\lambda,\mu}(\brho)~ 
            {\sigma}^{(1)}_{\lambda,\mu}(l,m)~Z_{\lambda-l}^{\abs{\mu}+m}(\bm v)\\
            % (2)
            &+\sum_{\lambda=l+1}^L\sum_{\mu=m+1}^{\min\set{-1,m-(l-\lambda)}}\gamma_{\lambda,\mu}(\brho)~ 
            {\sigma}^{(2)}_{\lambda,\mu}(l,m)~Z_{\lambda-l}^{-(m+\abs{\mu})}(\bm v)\\
            % (3)
            &-\sum_{\lambda=l+1}^L\sum_{\mu=-m-(\lambda-l)}^{-1} \gamma_{\lambda,\mu}(\brho)~ 
            {\sigma}^{(3)}_{\lambda,\mu}(l,-m)~Z_{\lambda-l}^{\abs{\mu}-m}(\bm v)\Bigg]\\
        % summand 3
        +\sum_{l=1}^L Z_l^0(\bm a)& \Bigg[\sum_{\lambda=l}^L\sum_{\mu=-(\lambda-l)}^{-1} \gamma_{\lambda,\mu}(\brho)~ 
        {\sigma}^{(1)}_{\lambda,\mu}(l,0)~Z_{\lambda-l}^{-\abs{\mu}}(\bm v)\Bigg]\\
        +Z_0^0(\bm a)& \Bigg[\sum_{\lambda=1}^L\sum_{\mu=-\lambda}^{-1} \gamma_{\lambda,\mu}(\brho)~ 
        {\sigma}^{(1)}_{\lambda,\mu}(0,0)~Z_{\lambda}^{-\abs{\mu}}(\bm v)\Bigg].
    \end{align*}

\subsection*{Translation of the coefficients}
Now, we have anything on hand to obtain the translated coefficients. Each summand of $\tau_{\bm v} \circ \cS_L(\bgam(\brho))$ is rearranged into a form $\sum_{l,m} Z_l^m \sum_{\lambda,\mu} \kappa(l,m,\lambda,\mu)$ so that we can put all parts together and define the translation of the coefficients as the sum over all corresponding $\kappa$.\\
The translation for $l\neq 0$ and $m>0$ is defined as
\begin{align}\label{eq:Coeffs_m>0}
    \begin{split}
        \hat{\tau}_{\bm v}\!\left(\gamma_{l,m}(\brho)\right) := 
        % mu > 0
            % (1)
            &\sum_{\lambda=l}^L\sum_{\mu=m}^{m-(l-\lambda)} \gamma_{\lambda,\mu}(\brho)~
            {\sigma}^{(1)}_{\lambda,\mu}(l,m)~Z_{\lambda-l}^{\mu-m}(\bm v)\\
            % (2)
            &+ \sum_{\lambda=l+1}^L\sum_{\mu=\max\set{1,m-(\lambda-l)}}^{m-1} \gamma_{\lambda,\mu}(\brho)~
            {\sigma}^{(2)}_{\lambda,\mu}(l,m)~Z_{\lambda-l}^{m-\mu}(\bm v)\\
            % (3)
            &+ \sum_{\lambda=l+1}^L\sum_{\mu=1}^{-m-(l-\lambda)}\gamma_{\lambda,\mu}(\brho)~
            {\sigma}^{(3)}_{\lambda,\mu}(l,-m)~Z_{\lambda-l}^{\mu+m}(\bm v)\\
        % mu = 0
        &+ \sum_{\lambda=l+m}^L \gamma_{\lambda,0}(\brho)~
        \left({\sigma}^{(2)}_{\lambda,0}(l,m)+{\sigma}^{(3)}_{\lambda,0}(l,-m)\right)~Z_{\lambda-l}^m(\bm v)\\
        % mu < 0
            % (1)
            &+ \sum_{\lambda=l}^L\sum_{\mu=-m-(\lambda-l)}^{-m-1} \gamma_{\lambda,\mu}(\brho)~
            {\sigma}^{(1)}_{\lambda,\mu}(l,-m)~Z_{\lambda-l}^{\mu+m}(\bm v)\\
            % (2)
            &- \sum_{\lambda=l+1}^L\sum_{\mu=-m+1}^{\min\set{-1,-m-(l-\lambda)}} \gamma_{\lambda,\mu}(\brho)~
            {\sigma}^{(2)}_{\lambda,\mu}(l,-m)~Z_{\lambda-l}^{-(\mu+m)}(\bm v)\\
            % (3)
            &+ \sum_{\lambda=l+1}^L\sum_{\mu=m-(\lambda-l)}^{-1} \gamma_{\lambda,\mu}(\brho)~
            {\sigma}^{(3)}_{\lambda,\mu}(l,m)~Z_{\lambda-l}^{\mu-m}(\bm v),
    \end{split}
\end{align}
for $l\neq 0$ and $m<0$ it is defined as
\begin{align}
\label{eq:Coeffs_m<0}
    \begin{split}
        \hat{\tau}_{\bm v}\big(\gamma_{l,m}(\brho)\big) := 
        % mu > 0
            % (1)
            &-\sum_{\lambda=l}^L\sum_{\mu=-m+1}^{-m-(l-\lambda)} \gamma_{\lambda,\mu}(\brho)~
            {\sigma}^{(1)}_{\lambda,\mu}(l,-m)~Z_{\lambda-l}^{-(\mu+m)}(\bm v)\\
            % (2)
            &+ \sum_{\lambda=l+1}^L\sum_{\mu=\max\set{1,-m-(\lambda-l)}}^{-m-1} \gamma_{\lambda,\mu}(\brho)~
            {\sigma}^{(2)}_{\lambda,\mu}(l,-m)~Z_{\lambda-l}^{\mu+m}(\bm v)\\
            % (3)
            &+ \sum_{\lambda=l+1}^L\sum_{\mu=1}^{m-(l-\lambda)} \gamma_{\lambda,\mu}(\brho)~
            {\sigma}^{(3)}_{\lambda,\mu}(l,m)~Z_{\lambda-l}^{-(\mu-m)}(\bm v)\\
        % mu = 0
        &+ \sum_{\lambda=l-m}^L \gamma_{\lambda,0}(\brho)~
        \left({\sigma}^{(2)}_{\lambda,0}(l,-m)+{\sigma}^{(3)}_{\lambda,0}(l,m)\right)~Z_{\lambda-l}^m(\bm v)\\ 
        % mu < 0
            % (1)
            &+ \sum_{\lambda=l}^L\sum_{\mu=m-(\lambda-l)}^{m} \gamma_{\lambda,\mu}(\brho)~
            {\sigma}^{(1)}_{\lambda,\mu}(l,m)~Z_{\lambda-l}^{m-\mu}(\bm v)\\
            % (2)
            &+ \sum_{\lambda=l+1}^L\sum_{\mu=m+1}^{\min\set{-1,m-(l-\lambda)}} \gamma_{\lambda,\mu}(\brho)~
            {\sigma}^{(2)}_{\lambda,\mu}(l,m)~Z_{\lambda-l}^{\mu-m}(\bm v)\\
            % (3)
            &- \sum_{\lambda=l+1}^L\sum_{\mu=-m-(\lambda-l)}^{-1} \gamma_{\lambda,\mu}(\brho)~
            {\sigma}^{(3)}_{\lambda,\mu}(l,-m)~Z_{\lambda-l}^{-(\mu+m)}(\bm v),
    \end{split}
\end{align}
for $l\neq 0$ and $m=0$ it is given by
\begin{align}\label{eq:Coeffs_m=0}
    \begin{split}
        \hat{\tau}_{\bm v}\big(\gamma_{l,0}(\brho)\big) := 
        % mu > 0
        &\sum_{\lambda=l}^L\sum_{\mu=1}^{\lambda-l} \gamma_{\lambda,\mu}(\brho)~
        {\sigma}^{(1)}_{\lambda,\mu}(l,0)~Z_{\lambda-l}^{\mu}(\bm v)\\
        % mu = 0
        &+ \sum_{\lambda=l}^L \gamma_{\lambda,0}(\brho)~
        {\sigma}^{(1)}_{\lambda,0}(l,0)~Z_{\lambda-l}^0(\bm v)\\
        % mu < 0
        &+ \sum_{\lambda=l}^L\sum_{\mu=-(\lambda-l)}^{-1} \gamma_{\lambda,\mu}(\brho)~
        {\sigma}^{(1)}_{\lambda,\mu}(l,0)~Z_{\lambda-l}^{\mu}(\bm v)\\
        % combined
        = &\sum_{\lambda=l}^L\sum_{\mu=-(\lambda-l)}^{\lambda-l} \gamma_{\lambda,\mu}(\brho)~
        {\sigma}^{(1)}_{\lambda,\mu}(l,0)~Z_{\lambda-l}^{\mu}(\bm v),
    \end{split}
\end{align}
and finally for $l = m = 0$ it is defined as
\begin{align}\label{eq:Coeffs_l=0}
    \begin{split}
        \hat{\tau}_{\bm v}\big(\gamma_{0,0}(\brho)\big) := 
        % mu > 0
        &\sum_{\lambda=1}^L\sum_{\mu=1}^{\lambda} \gamma_{\lambda,\mu}(\brho)~
        {\sigma}^{(1)}_{\lambda,\mu}(0,0)~Z_{\lambda}^{\mu}(\bm v)\\
        % mu = 0
        &+ \sum_{\lambda=0}^L \gamma_{\lambda,0}(\brho)~ 
        {\sigma}^{(1)}_{\lambda,0}(0,0)~Z_{\lambda}^0(\bm v)\\
        % mu < 0
        &+ \sum_{\lambda=1}^L\sum_{\mu=-\lambda}^{-1} \gamma_{\lambda,\mu}(\brho)~
        {\sigma}^{(1)}_{\lambda,\mu}(0,0)~Z_{\lambda}^{\mu}(\bm v)\\
        % combined
        = &\sum_{\lambda=0}^L\sum_{\mu=-\lambda}^{\lambda} \gamma_{\lambda,\mu}(\brho)~
        {\sigma}^{(1)}_{\lambda,\mu}(0,0) Z_{\lambda}^{\mu}(\bm v),
    \end{split}
\end{align}
which is equal to~\eqref{eq:Coeffs_m=0} with $l=0$.

With these definitions we finally obtain the operator $\signatur{\hat{\tau}_{\bm v}}{\IR^{(L+1)^2}}{\IR^{(L+1)^2}}$ such that
\begin{align*}
        \tau_{\bm v} \circ \mathcal{S}_L(\bgam(\brho)) 
        = \mathcal{S}_L \circ \hat{\tau}_{\bm v}\big(\bgam(\brho)\big).
    \end{align*}
\end{proof}